\documentclass[aps,prb,notitlepage,nofootinbib,floatfix,twocolumn,superscriptaddress, longbibliography]{revtex4-2}
\usepackage{dsfont}
\usepackage{amsmath,amsthm}
\usepackage{amssymb}
\usepackage{mathtools}
\usepackage[tight]{subfigure}
\usepackage{enumitem}
\usepackage{soul}
\usepackage{cancel}
\usepackage{multirow}
\usepackage{tikz}
\usepackage{pifont} 
\usepackage{xspace}
\usepackage{comment}
\usepackage{physics}
\usepackage{bbold}
\usepackage{graphicx}
\usepackage{dcolumn}
\usepackage{bm}
\usepackage[breaklinks]{hyperref}
\hypersetup{colorlinks=true, linkcolor=blue, citecolor=blue, filecolor=blue, urlcolor=blue}
\usepackage{ulem}

\theoremstyle{definition}
\newtheorem{definition}{Definition}[section]
\newtheorem{theorem}{Theorem}
\newtheorem*{theorem*}{Theorem}

\newtheorem*{lemma*}{Lemma}
\newtheorem*{proposition*}{Proposition}

\newtheorem{corollary}{Corollary}[section]

\newtheorem*{definition*}{Definition}  
\newcommand{\CC}[1]{\textcolor{blue}{(CC) #1}}

\begin{document}
\title{Stabilizer Entanglement Enhances Magic Injection}

\author{Zong-Yue Hou}
\affiliation{School of Physics, Peking University, Beijing 100871, China}

\author{ChunJun Cao}
\email{cjcao@vt.edu}
\affiliation{Department of Physics, Virginia Tech, Blacksburg, VA, USA 24061}

\author{Zhi-Cheng Yang}
\email{zcyang19@pku.edu.cn}
\affiliation{School of Physics, Peking University, Beijing 100871, China}
\affiliation{Center for High Energy Physics, Peking University, Beijing 100871, China}

\date{\today}

\begin{abstract}
Non-stabilizerness is a key resource for fault-tolerant quantum computation, yet its interplay with entanglement in dynamical settings remains underexplored. We address this by analyzing a well-controlled, analytically tractable setup, where we show that entanglement acts as a conduit that teleports magic across the system, thereby enhancing magic injection. Using exact calculations, we prove that when a Haar-random unitary $U_A$ is applied to a subsystem $A$ of an entangled stabilizer state, the total injected magic increases with the entanglement between $A$ and its complement. More generally, for any unitary $U_A$, we show that this enhancement is maximized when $A$ is maximally entangled with its complement, in which case the total injected magic is exactly given by the unitary stabilizer Rényi entropy we introduce. This quantity provides both a directly computable measure of unitary magic and a lower bound on the minimum number of $T$ gates required to synthesize $U_A$. We further extend our analysis to tripartite stabilizer entanglement, non-stabilizer entanglement, and magic injection via shallow-depth brickwork circuits, finding that the qualitative picture remains unchanged.
  
\end{abstract}

\maketitle

{\it Introduction.-} Understanding what distinguishes quantum systems from classical ones has been a challenging and multi-faceted endeavor. Much of the study of quantumness in physics has focused on the idea of quantum entanglement, which played a central role in characterizing topological order~\cite{Hamma_2005,Levin_2006,PreskillKitaev}, dynamically induced phases~\cite{Skinner_2019}, quantum chaos, and even the emergence of spacetime~\cite{Ryu_2006}. However, entanglement being a quantum correlation, is just one facet of quantumness. Quantum advantage, for example, is concerned with the hardness of simulating quantum systems on a classical computer. This facet is not captured by entanglement alone as many highly entangled systems prepared using Clifford operations violate Bell's inequalities but can be efficiently simulated classically~\cite{gottesmanknill,Aaronson_2004}. The notion of classical hardness therefore constitutes a second layer of quantumness and is intimately connected to non-stabilizerness or magic~\cite{HowardCampbell,Veitch_2014,Heinrich_2019,magicMC,Howard_2014,Bravyi_2005}.

Like entanglement, magic is an important but distinct quantum resource for realizing fault-tolerant \textit{universal} quantum computation~\cite{Bravyi_2005,Bravyi_2012}. 
The importance of non-stabilizerness has been largely overlooked in quantum many-body physics until recently, where a flurry of activities have been dedicated to study its total extent in various systems, which has implication for resource estimation in state preparation and quantum simulation~\cite{Sarkar_2020,white2020manahaarrandomstates,White:2020zoz,Liu_2022,Paulispec_haarrand, chen2024magic, bravyi2019simulation, PRXQuantum.6.020324}, in characterizing phase transitions~\cite{Niroula:2023meg,magic_transition}, in quantum dynamics~\cite{zhou2020single, manadynamics,SREdynamics,zhang2024quantummagicdynamicsrandom,monitored_magic_transit, wang2025magic}, in improved quantum simulation with tensor networks~\cite{clifford_disentangling,clifford_disentangling2, PhysRevLett.133.230601, dowling2025bridging, andreadakis2025exact}, and in quantum gravity~\cite{Cao:2023mzo,grav_magical}. Arguably, the most fascinating perspective comes from the interplay between magic and entanglement \cite{magic_spec,grav_magical,stabent25,iannotti2025,magicalent,gu2024,monitored_magic_transit} as both are needed for quantum advantage. Recent work revealed that magic in highly entangled systems is similar to the scrambling/encoding of information in quantum error correcting codes where it quickly delocalizes and can only be found when one can access a large part of the system \cite{Bao_2022,zhang2024quantummagicdynamicsrandom, wei2024noise}. This asymptotic behavior is consistent with the observation that subsystem magic quickly rises and decays to almost zero as the system thermalizes~\cite{manadynamics}.

Despite recent progress, the fundamental mechanisms governing magic dynamics, particularly the interplay between magic injection and entanglement, remain poorly understood. 
In generic quantum circuits and Hamiltonian dynamics, entanglement and magic dynamics are interwined, as each time step of evolution simultaneously injects magic and increases entanglement~\cite{SREnonintegrable,magicspread_randomcircuit}. To disentangle their roles, it is essential to study a controlled setting where one can isolate their contributions.  In this work, by constructing such a setup, we aim to address the following question: how does entanglement affect the injection of magic into a quantum system?

We show that preexisting stabilizer entanglement (i.e. entanglement of a nonmagical, stabilizer state) can enhance magic injection via gate teleportation. More explicitly, using exact calculations, we prove that when a Haar-random unitary $U_A$ is applied to a subsystem $A$ of an entangled stabilizer state, the total injected magic increases with the entanglement between $A$ and its complement. More generally, for any unitary $U_A$, we show that this enhancement is maximized when $A$ is maximally entangled with its complement, in which case the total injected magic is exactly given by the unitary stabilizer Rényi entropy we introduce. This quantity provides both a directly computable measure of unitary magic and a lower bound on the minimum number of $T$ gates (defined as $T=e^{i\frac{\pi}{8}Z}$) required to synthesize $U_A$, also known as $T$-count. We further extend our analysis to tripartite stabilizer entanglement, non-stabilizer entanglement, and magic injection via shallow-depth brickwork circuits, finding that the qualitative picture remains unchanged.

{\it Quantum magic and its measure.-} 
Clifford operations consist of unitaries that map Pauli operators to Pauli operators, the preparation of Pauli eigenstates, and the measurement in the Pauli basis among other things. Their simplicity around Pauli operators allows them to be efficiently simulable on a classical computer. The states that can be prepared with Clifford operations are called stabilizer states. A stabilizer pure state $|\psi\rangle$ is one where $s|\psi\rangle=|\psi\rangle$ for all $s\in S$ where $S$ is a size-$2^N$ Abelian subgroup of the Pauli group $P_N$ over $N$ qubits. It can be shown that $\rho=|\psi\rangle\langle\psi|=\frac{1}{2^N}\sum_{s\in S}s$ --- a stabilizer pure state is an equal superposition of the $2^N$ stabilizer group elements. 

Quantum magic quantifies the amount of ``non-Clifford" or ``non-stabilizer'' resources needed in a quantum state or process. In this work, we focus on a computable measure of magic in quantum states called {\it linear stabilizer entropy}~\cite{SRE} defined as
\begin{equation}
Y^{\rm lin}(\rho) := 1- \frac{1}{2^{N}}\sum_{P\in P_N}{\rm tr}(P \rho)^4. 
\end{equation} 
As $\tr(P\rho)$ computes the expansion coefficients of the state in the Pauli basis, also known as the Pauli spectrum, $1-Y^{\rm lin}(\rho)$ is simply the second moment of the Pauli spectrum of $\rho$, which characterizes how delocalized the expectation values of $P$ are over all possible Pauli strings compared to the stabilizer spectrum which contains precisely $2^N$ ones and $4^N-2^N$ zeros.

The linear stabilizer entropy satisfies the following properties: (i) faithfulness: let \(\text{STAB}_N\) denote the set of all \(N\)-qubit pure stabilizer states, then \(Y^{\rm lin}(\rho)=0\) iff \(\rho \in \text{STAB}_N\), otherwise \(Y^{\rm lin}(\rho)>0\); (ii) invariant under Clifford unitaries \(U_C\): \(Y^{\rm lin}(U_C\rho U^{\dagger}_C)=Y^{\rm lin}(\rho)\); 
(iii) monotonicity: $Y^{\rm lin}$ is a strong pure-state monotone under free operations \cite{Monotone}. Hence, $Y^{\rm lin}(\rho)$ serves as a good measure for quantum magic from the point of view of magic-state resource theory, which increases and approaches unity as magic increases. Notice that $Y^{\rm lin}(\rho)$ is also closely related to the second stabilizer R\'enyi entropy (SRE) $M_2(\rho)= - {\rm log} (1-Y^{\rm lin}(\rho))$, another widely used magic measure~\cite{SRE}.

\begin{figure}[!tb]
    \centering
    \includegraphics[width=0.95\linewidth]{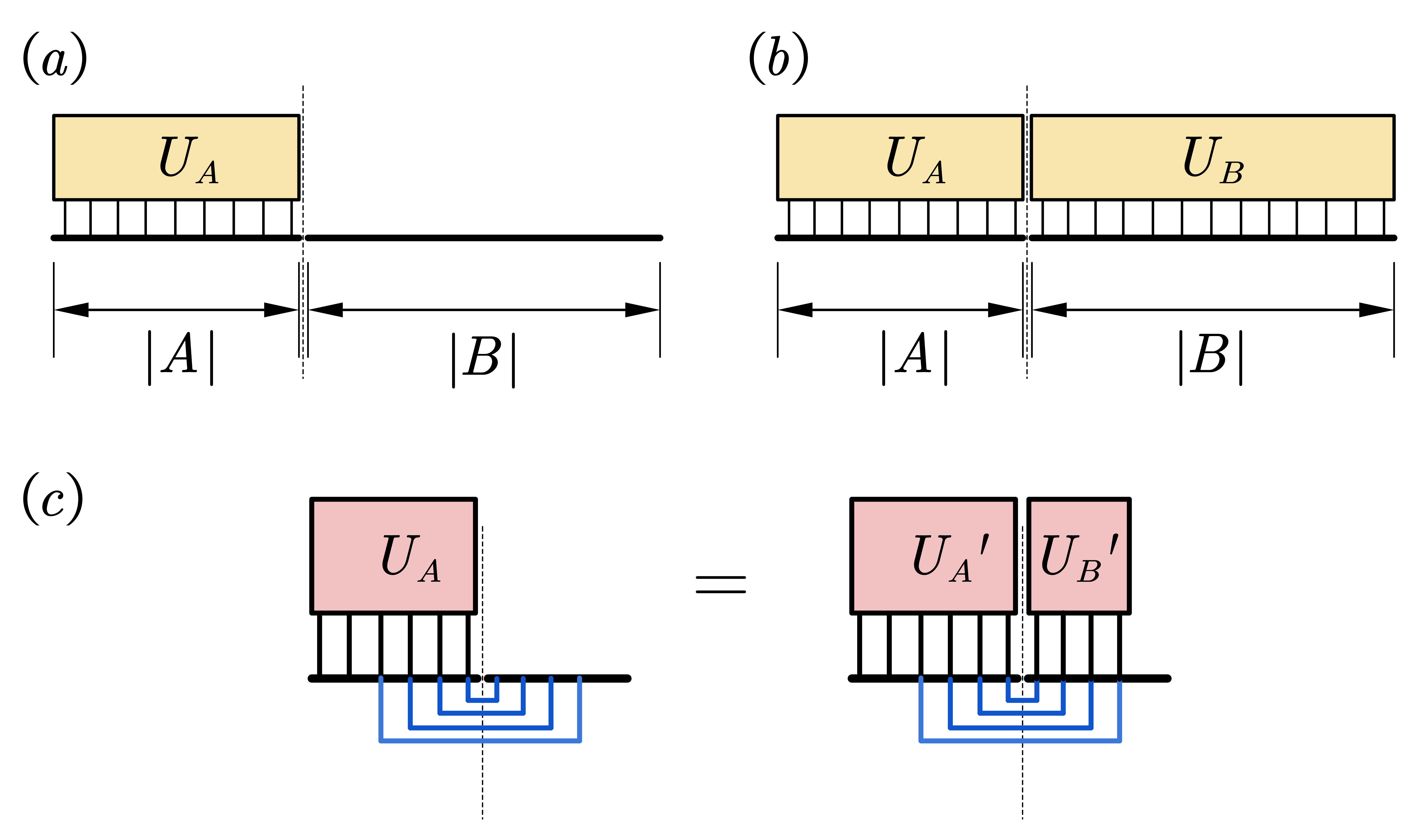}
    \caption{Schematics of the setup considered in this work. The initial states are bipartite stabilizer pure states. (a) A Haar random unitary injects magic by acting on a subregion $A$. (b) A factorized unitary $U=U_A \otimes U_B$, where $U_A$ and $U_B$ are independently Haar-random, acts on the initial stabilizer state. (c) Schematic of gate teleportation. Because of the Bell pairs between $A$ and $B$, \(U_A\) acting on $A$ is equivalent to \({U'_A}\) and \({U'_B}\) acting on $A$ and $B$.}
    \label{fig:Fig1}
\end{figure}

{\it Stabilizer entanglement enhances magic injection.-} To see how entanglement impacts magic injection, consider the setup illustrated in Fig.~\ref{fig:Fig1}(a). We take an initial stabilizer pure state $\rho$, partitioned into subregion $A$ and its complement $B$. A Haar random unitary $U_A$ is then applied to subregion $A$, and we are interested in computing the linear stabilizer entropy of the resulting state $\sigma=U_A|\psi\rangle \langle \psi| U_A^\dagger$, averaged over Haar random unitaries. This setup preserves the entanglement spectrum between $A$ and $B$ and models the process where magic is locally injected by a non-Clifford unitary gate to a stabilizer state when $|A| \ll |B|$. We will show in SM~\cite{SM} that key qualitative features of our results also persist even when considering more structured unitaries within subregion $A$.

We then compute Haar-averaged the linear stabilizer entropy $Y^{\rm lin}(U_A|\psi\rangle \langle \psi|U_A^\dagger)$ exactly. To leading order,
\begin{eqnarray}
    \overline{Y^{\rm lin}}&=& \mathbb{E}_{U_A} \ Y^{\rm lin}\left(U_A|\psi\rangle \langle\psi|\ U_A^\dagger\right)  \nonumber \\
    &=&1-4\cdot {{2}^{-|A|-E}}\left[1+O({2}^{-{|A|-E}})\right],
\label{eq:Y_bi}
\end{eqnarray}
 where $|A|\gg 1$ and $E=-\rm{tr}(\rho_A\text{log}\rho_A)$ is the amount of entanglement between $A$ and $B$, which is an integer for a stabilizer state.
 An exact closed-form expression for $\overline{Y^{\rm lin}}$ at arbitrary $|A|$ and $E$ is derived in the Supplemental Material (SM)~\cite{SM}. It turns out that the leading-order analytical formula already provides an excellent approximation for moderate system sizes and entanglement: numerical simulations show deviations become negligible once $|A|+E \gtrsim 4$ [Fig.~\ref{fig:Fig2}(a)]. Consistently, we find that the Haar-averaged second SRE $\overline{M_2}$ is nearly indistinguishable from $-\log(1-\overline{Y^{\rm lin}})$, confirming that the average in Eq.~(\ref{eq:Y_bi}) is typical with vanishing fluctuations~\cite{SM}. Thus we have
 \begin{equation}
 \overline{M_2} = |A|+E+ 2 + O(2^{-|A|-E}).
 \label{eq:M2}
 \end{equation}
 Notice that $\overline{M_2}$ is extensive, as expected.
 Eqs.~(\ref{eq:Y_bi}) and (\ref{eq:M2}) establish a \textit{quantitative} relation between the efficiency of local magic injection and the amount of stabilizer entanglement contained in the target state. 
 For a fixed subregion $A$, $\overline{Y^{\rm lin}}(\rho)$ approaches $1$ exponentially with the amount of entanglement $E$ between $A$ and its complement, and correspondingly the average second SRE increases linearly with $E$, indicating that global magic injection on $AB$ is enhanced by preexisting entanglement. 
 Similar behavior arises when $U_A$ is implemented by a brickwork circuit rather than a Haar-random unitary, although in that case the total injected magic is limited by the circuit depth (see SM~\cite{SM}).

Intuitively, the enhancement of magic injection by stabilizer entanglement as revealed by Eqs.~(\ref{eq:Y_bi}) and~(\ref{eq:M2})  can be understood as a consequence of gate `teleportation' mediated by Bell pairs spanning subsystems $A$ and $B$ [see Fig.~\ref{fig:Fig1}(c)]. As a simple example, consider a $2N$-qubit maximally entangled state between $A$ and $B$: $|{\rm Bell}\rangle_{AB}^{\otimes N}$. Then any unitary acting on half of the Bell pairs in $A$ will teleport to its partner in $B$ since $U \otimes I\ |\mathrm{Bell}\rangle_{AB}^{\otimes N}= I \otimes U^T\ |\mathrm{Bell}\rangle_{AB}^{\otimes N}$.
It has been shown that any bipartite stabilizer state $|\psi\rangle_{AB}$ can be converted into a collection of single-qubit states (say $|+\rangle$) and Bell pairs through local Clifford unitaries~\cite{bravyi2006ghz}:
\begin{equation}
    U_A^{\rm Cliff}\otimes U^{\rm Cliff}_B|\psi\rangle_{AB}=|+\rangle^{\otimes f_A}|\mathrm{Bell}\rangle^{\otimes E}_{AB}|+\rangle^{\otimes f_B}.
\end{equation} 
Hence for a stabilizer state that has entanglement, one can inject magic on $B$ through $A$ where magic flows to $B$ via the ``entanglement conduit'' formed by the $E$ Bell pairs. This simple picture also offers an intuitive explanation for the observed dynamics of local magic as a system thermalizes \cite{manadynamics}---while time evolution increases local magic initially when the spins are unentangled, this local magic quickly spreads through the conduit as entanglement also grows in the system.

\begin{figure}[!tb]
    \centering
    \includegraphics[width=1\linewidth]{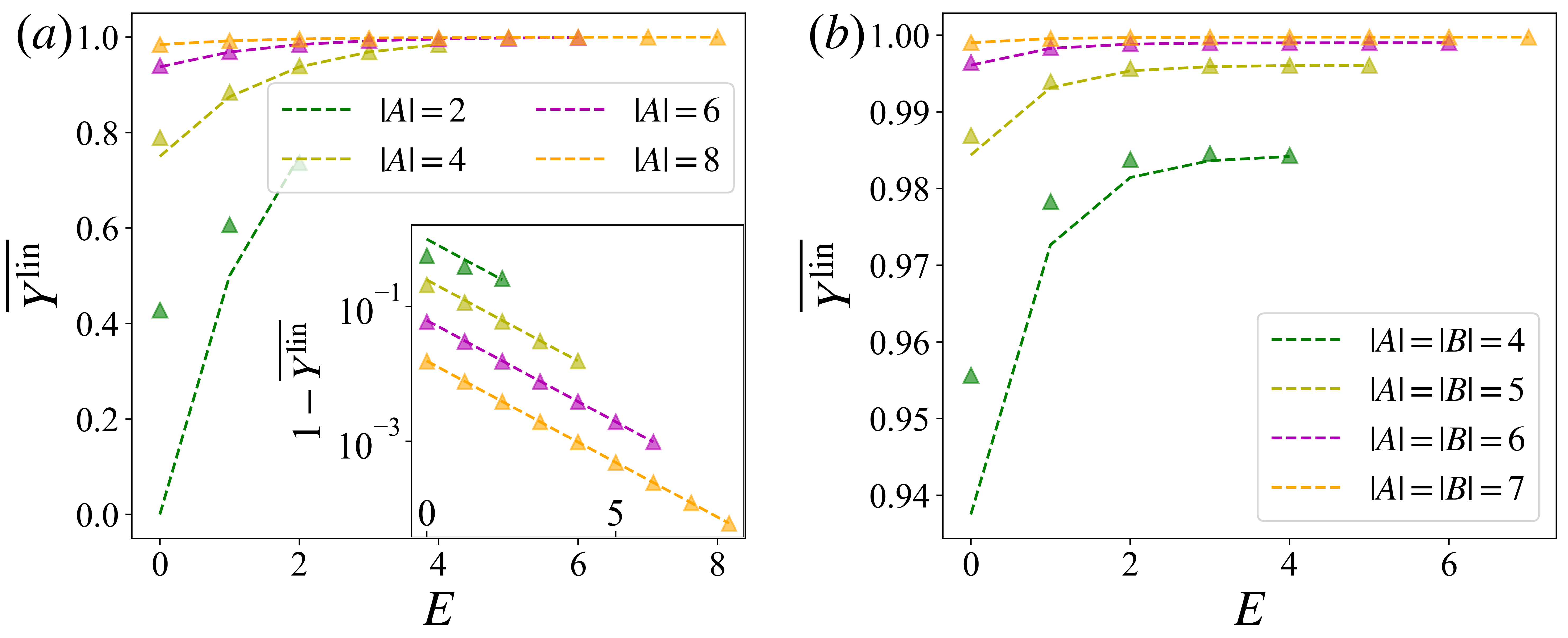}
    \caption{Numerical simulations of $\overline{Y^{\rm lin}}$ for the setup depicted in (a) Fig.~\ref{fig:Fig1}(a), and (b) Fig.~\ref{fig:Fig1}(b), respectively. We consider stabilizer initial states with varying system sizes and bipartite entanglement $E$. The analytical results Eqs.~(\ref{eq:Y_bi}) and (\ref{eq:UaUb}) are shown as dashed lines.}
    \label{fig:Fig2}
\end{figure}

A further manifestation of entanglement-enhanced magic injection is revealed by applying a unitary that is a tensor product of local unitaries acting on subsytems $A$ and $B$: $U=U_A \otimes U_B$, as illustrated in Fig.~\ref{fig:Fig1}(b). We compute the linear stabilizer entropy of the resulting state, averaged over Haar random $U_A$ and $U_B$. In the limit where $|A|\gg 1$ and $|B| \gg 1$, we find~\cite{SM}
\begin{eqnarray}
\overline{Y^{\rm lin}}&=& \mathbb{E}_{U_A}\mathbb{E}_{U_B} Y^{\rm lin}\left(U_AU_B|\psi\rangle \langle \psi |U_A^\dagger U_B^\dagger\right)  \nonumber \\
&=& 1-4\cdot {2}^{-N}\left[1+3\cdot 2^{-2E}+O({{2}^{-|A|-E}}, {{2}^{-|B|-E}})\right],   \nonumber \\
\label{eq:UaUb}
\end{eqnarray}
where $N=|A|+|B|$ is the total number of qubits.  First of all, Eq.~(\ref{eq:UaUb}) is consistent with Eq.~(\ref{eq:Y_bi}) and our teleportation picture upon taking $|A|=|B|=E$.
Furthermore, notice that the average value of $\overline{Y^{\rm lin}}$ for Haar random states is given by $1-4\cdot 2^{-N}$. Eq.~(\ref{eq:UaUb}) shows that the gap between $\overline{Y^{\rm lin}}$ for a stabilizer state acted upon by $U_A \otimes U_B$  and that for a Haar random state decays exponentially with the amount of entanglement. Therefore, despite the applied unitary being a tensor product of local unitaries on subsystems $A$ and $B$, from the perspective of magic, its effect on magic injection closely resembles that of a genuinely global Haar-random unitary on $AB$, provided that the initial state is sufficiently entangled. Interestingly, the amount of entanglement necessary for such indistinguishability from Haar random magic is independent of system size. Fig.~\ref{fig:Fig2}(b) shows numerical simulations of Haar-averaged $U_A$ and $U_B$, which again closely match Eq.~(\ref{eq:UaUb}) with subleading corrections already negligible at moderate $|A|$ and $E$. Remarkably, $\overline{Y^{\rm lin}}$ saturates to the Haar-random value once $E \approx 3$, independent of system size. 

In SM, we generalize Eqs.~(\ref{eq:Y_bi}) and (\ref{eq:UaUb}) to tripartite states, where similar results are obtained~\cite{SM}.

{\it T count and unitary SRE.-} 
The preceding results establish that entanglement enhances magic injection, but they also raise a natural question: what sets the ultimate limit of this enhancement? From a resource-theoretic perspective, all non-Clifford resources originate from the applied unitary $U_A$, so the total injected magic must be upper bounded by the non-Clifford content of $U_A$ itself. We begin by introducing the unitary stabilizer Rényi entropy of a given unitary, which provides a direct bridge between our previous results and the non-Clifford resources contained in $U$.

\begin{definition*}
    Let $U$ be an $N$-qubit unitary, then the $\alpha-$th unitary stabilizer R\'enyi entropy of $U$ is defined as
\begin{equation}
        H_{\alpha}(U):=\frac{1}{1-\alpha}\log \Big(\frac{1}{2^{2N}}\sum_{P_{i},P_{j}}\Big[\frac{{\text{tr}}(P_{i}UP_{j}U^{\dagger})}{2^{N}}\Big]^{2\alpha}\Big)
\label{eq:hu}
\end{equation}
where $P_i \in P_N$. 
\end{definition*}
 To gain some intuition on the above quantity, notice that $U P_i U^\dagger = \sum c_{ij} P_j$ for a single Pauli string $P_i$, where $c_{ij}= 2^{-N} {\rm tr}(P_j U P_i U^\dagger)$, and $\sum_{j} |c_{ij}|^2 =1$. Thus Eq.~(\ref{eq:hu}) is essentially the logarithm of $2^{-2N} \sum_{ij} |c_{ij}|^{2\alpha}$, where the additional factor of $2^{-2N}$ is due to the summation over all $P_i$. The definition clearly parallels the state SRE, but now for the unitary directly, which basically characterizes the similarity between $U$ and a Clifford unitary for which $c_{ij}=\pm 1$ for one specific $j$. 

We prove that this measure establishes a lower bound on the $T$-count, i.e. the minimal number of $T$ gates needed to synthesize a unitary $U$. More precisely,
\begin{theorem}
For any  $\alpha\geq 0$,
   \begin{equation}
t(U) \geq v(U) \geq H_{\alpha}(U),
\label{eq:ineq}
\end{equation} 
where the unitary stabilizer nullity is defined as $v(U):= 2N-{\rm log} |s(U)|$, where $s(U):=\{P_i\in P_N| U P_i U^\dagger = \pm P_j \}$~\cite{Jiang_2023}.  
\end{theorem}

If the total amount of magic of the unitary is given by the unitary SRE above, what initial stabilizer state $|\psi\rangle$, if it exists, can be used to extract all the magic hidden in $U_A$? The answer is simply a state that is maximally entangled between $A$ and an ancillary system $B$ of the same size. 
\begin{theorem}
The unitary SRE of $U$ is equal to the state SRE of $U$'s Choi state $|U\rangle:=U_A\otimes I_B|\rm Bell\rangle_{AB}^{\otimes N}$: 
\begin{equation}
    H_{\alpha}(U)=M_{\alpha}(|U\rangle).
    \label{eq:lemma}
\end{equation}
\end{theorem}
Notice that the right hand side of Eq.~(\ref{eq:lemma}) is precisely the kind of setup we consider in this work, namely, a unitary acts on a stabilizer initial state consisting of $N$ Bell pairs shared between subsystem $A$ and its complement.
Explicitly, the linear stabilizer entropy of acting $U$ on subsystem $A$ when the initial stabilizer state has $|A|=E\leq |B|$ is $Y^{\rm lin}=1-\exp(-H_2(U))$.

The above theorem also provides an intuitive interpretation of the unitary SRE we introduce --- if we unfold a unitary operator into its dual state using the channel-state duality, then the magic of a unitary operator is precisely equal to that of its dual (Choi) state. Detailed proofs of Theorem 1 and 2 are given in the SM~\cite{SM}.

\begin{figure}[!tb]
    \centering
    \includegraphics[width=1\linewidth]{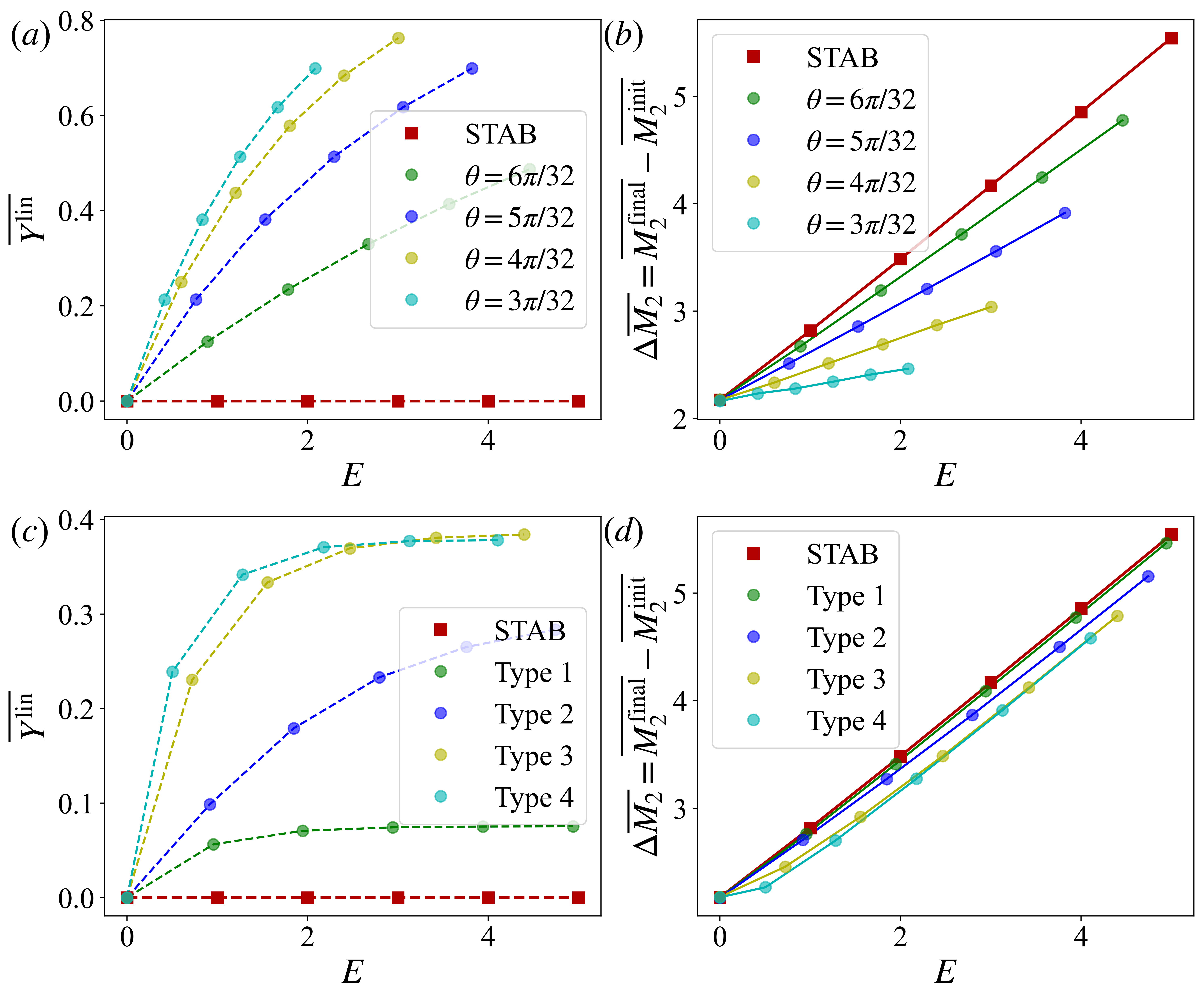}
    \caption{Numerical simulations of magic injection to initial states with non-stabilizer entanglement. (a)\&(b): imperfect Bell states; (c)\&(d): states with engineered entanglement spectra (see text for a precise description). In (a)\&(c), we plot the amount of preexisting magic in the initial states, while in (b)\&(d) we plot the amount of injected magic: $\Delta \overline{M_2} = \overline{M_2^{\rm final}}-\overline{M_2^{\rm init}}$.}
    \label{fig:Fig3}
\end{figure}

{\it Non-stabilizer entanglement.-} So far, we have focused on initial states with only stabilizer entanglement. However, realistic dynamics almost always contains magic in the initial state.  Does entanglement still enhance magic injection when the initial state is non-stabilizer? Our results indicate that, at leading order, the answer is yes. Consider an entangled state on $AB$ with magic injected by $U_A$. Because any local non-stabilizerness in $A$ is washed out by averaging over $U_A$, it suffices to focus on states where the non-stabilizer content resides entirely in the correlations between $A$ and $B$.
Such a state is known to have non-local magic, which is non-trivial as long as the entanglement spectrum is non-flat~\cite{grav_magical}. The simplest example of such a state consists of $k$ imperfect Bell pairs shared between $A$ and $B$: $|\psi\rangle = |0\rangle^{\otimes f_A} |\phi_{\theta}\rangle_{AB}^{\otimes k}$, where $|\phi_{\theta}\rangle = \cos\theta|00\rangle + \sin\theta|11\rangle $. 
Since $k$ copies of $|\phi_{\theta}\rangle$ each with entanglement $S$ can be unitarily distilled into $E=Sk-O(\sqrt{k})$ number of perfect Bell states tensoring a magic state on the $O(\sqrt{k})$ remaining qubits to good approximation~\cite{Lo}, our previous analysis with stabilizer initial state again applies to leading order. That is, magic injection by $U_A$ will increase in the same way as one increases the total entanglement $E$. Interestingly, this simple dependence also appears to extend to other states with non-local magic. For example, initial states with non-local magic can be parametrized in the form $|\psi\rangle = |0\rangle^{\otimes f_A} \sum_{i=1}^{2^k} \lambda_i |i\rangle_A \otimes |i\rangle_B$, where $|i\rangle$ denotes a computational basis state of $k$ qubits. We choose four arbitrary entanglement spectra: (i) $\lambda_i\propto e^{-i/2^k}$; (ii) $\lambda_i\propto i$; (iii) $\lambda_i\propto i^2$; (iv) $\lambda_i\propto i^3$ and examine the magic growth numerically.

However, more care is needed going beyond leading order. It is clear that there is non-stabilizerness dependence even in the simplest example with imperfect Bell states at the order of $O(\sqrt{k})$. Furthermore, because initial states can contain magic, the remaining magic capacity for further injection is reduced, i.e., the change of total magic per magic state injection is smaller compared to stabilizer entanglement.  We numerically calculate both the preexisting magic in the initial states and the injected magic by $U_A$, as quantified by the increase in the second SRE: $\Delta \overline{M_2}=\overline{M_2^{\rm final}}-\overline{M_2^{\rm init}}$. As shown in Fig.~\ref{fig:Fig3}(a,c), $\overline{Y^{\rm lin}}$ of both classes of initial states grow with entanglement, indicating that on average the initial state contains more magic when more entangled. 
The {\it injected} magic $\Delta {\overline{M_2}}$ also increases with the amount of entanglement [Fig.~\ref{fig:Fig3}(b,d)]. This suggests that even non-stabilizer entanglement enhances the effect of magic injection, albeit at a much smaller rate. Indeed, as seen from Fig.~\ref{fig:Fig3}, the amount of injected magic $\Delta \overline{M_2}$ becomes smaller as the initial state becomes more magical at the same $E$.

{\it Discussion.-} In this work, we demonstrate that entanglement enhances the effect of locally injected magic. We showed this in two ways --- by locally injecting magic using a Haar random unitary, and by directly computing the upper limit of maximally extractable magic using unitary SRE. The former shows a precise increase of magic as a function of entanglement while the latter proves that the upper limit is attained when the initial state is maximally entangled. We further generalize our results to tripartite entanglement and non-stabilizer entanglement.  As any stabilizer state is equivalent to the tensor product of Bell states, GHZ states and unentangled ancilla through local Clifford rotations, our results apply generally to any bipartite and tripartite stabilizer state. On the resource-theoretic front, we define a new computable measure $H_{\alpha}(U)$ for estimating $t(U)$, which leads to a natural interpretation of entanglement-enhanced magic injection as finding an initial stabilizer state that most efficiently extract the magic contained in a given unitary.
Furthermore, it does not rely on optimization and can be computed far more efficiently via tensor network methods, unlike unitary or state nullity. This makes $H_{\alpha}$ a more computable unitary magic measure and a powerful estimate of unitary magic.

Several open questions remain. Our analysis of imperfect Bell-pair states suggests that the average amount of injected magic decreases as the target state itself becomes more magical. This points to a possible notion of an intrinsic magic capacity for a quantum state, determined jointly by its entanglement and existing magic. A systematic extension of our results to generic non-stabilizer states would be needed to make this precise. Another direction concerns replacing the Haar-random $U_A$ with a structured brickwork circuit: what minimal circuit depth is required for the injected magic to match that of a Haar-random unitary? This question is likely connected to the growth of complexity in local random unitary circuits as a function of depth.

{\it Acknowledgement.-} We thank Zi-Wen Liu, Jong-Yeon Lee, Cheng Wang, Yingfei Gu, and Yuzhen Zhang for helpful discussions, and Xhek Turkeshi for useful comments on the manuscript. We especially thank Daniele Iannotti, Gianluca Esposito, Lorenzo Campos Venuti, and Alioscia Hamma for sharing their note on obtaining the exact expression of average stabilizer purity at a fixed amount of entanglement. Z.-C.Y. is supported by Grant No. 12375027 from the National Natural Science Foundation of China. Numerical simulations were performed on the High-performance Computing Platform of Peking University.

\bibliography{ref}

\clearpage
\onecolumngrid
\appendix 

\subsection*{Supplemental Material for ``Stabilizer Entanglement Enhances Magic Injection"}\label{sec:sup}

\section{Additional numerical results on the second stabilizer R\'enyi entropy $\overline{M_2}$}

In this section, we provide additional numerical results on the Haar-average value of the second stabilizer R\'enyi entropy $\overline{M_2}$, for the setup depicted in Fig.~\ref{fig:Fig1}(a) of the main text. We find that the the average $\overline{M_2}$ is extremely close to $-{\rm log}(1-\overline{Y^{\rm lin}})$, see Fig.~\ref{fig:SM1}. This indicates that the average $\overline{Y^{\rm lin}}$ is also typical, as sample-to-sample fluctuations are small.

\begin{figure}[h]
    \centering
    \includegraphics[width=0.6\linewidth]{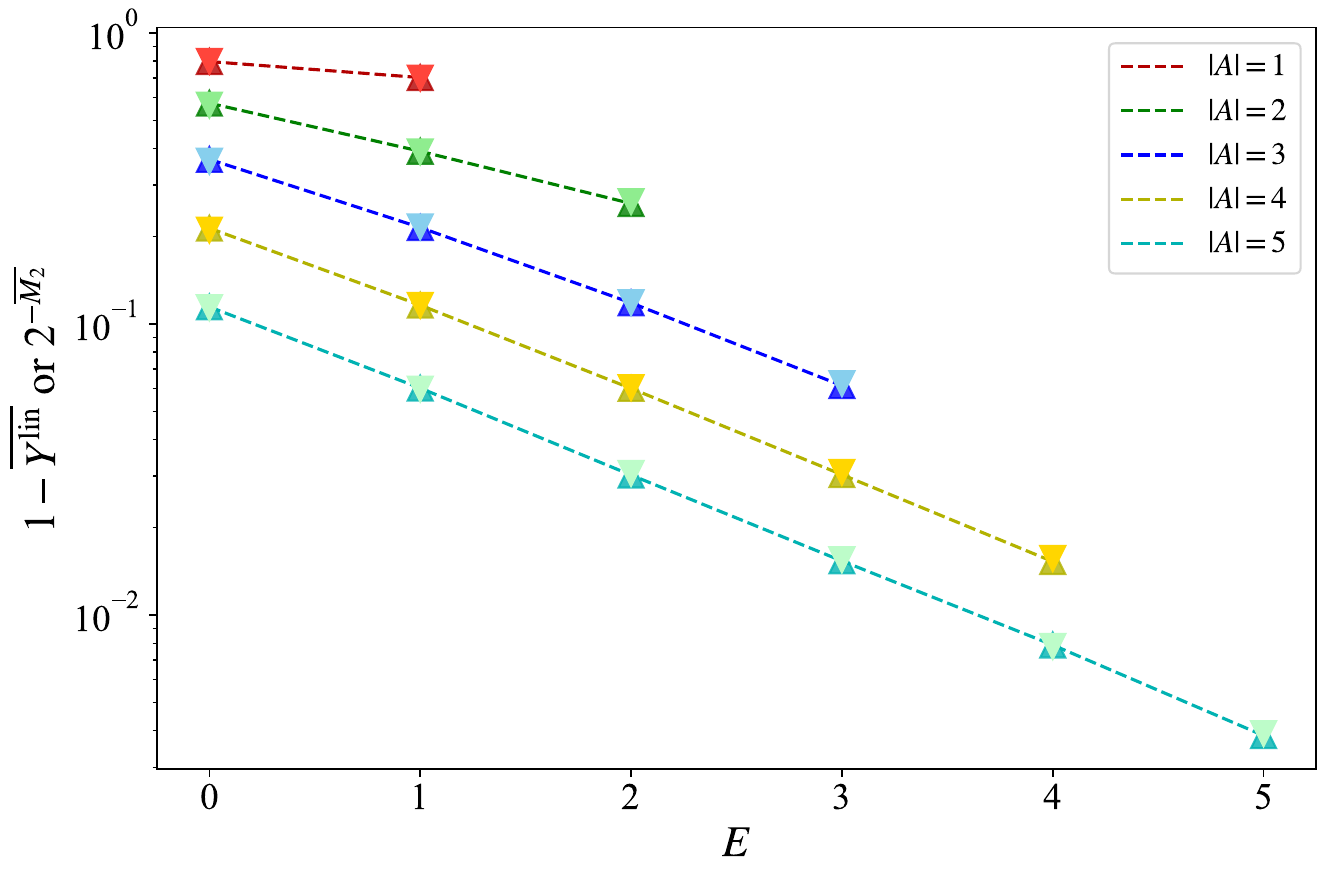}
    \caption{Numerical simulations of $\overline{Y^{\rm lin}}$ and $\overline{M_2}$ for the setup depicted in Fig.1(a) of the main text. $1-\overline{Y^{\rm lin}}$ is shown as upper triangles and $2^{-\overline{M_2}}$ is shown as lower triangles.}
    \label{fig:SM1}
\end{figure}


\section{Tripartite entanglement} \label{sec:tri}

In this section, we extend the main results presented in the main text to tripartite entanglement. Consider a stabilizer state shared among three parts $A$, $B$ and $C$. Since any tripartite stabilizer state can be transformed via local unitaries acting on $A$, $B$, and $C$ into a collection of GHZ states, Bell states, and single-qubit states~\cite{bravyi2006ghz}, we can, without loss of generality, consider initial states as depicted in Fig.~\ref{fig:SM2}(a). Specifically, the initial state consists of $g$ GHZ states shared among $A$, $B$, and $C$; $b_{AB}$, $b_{AC}$, and $b_{BC}$ Bell states shared between subsystems $AB$, $AC$, and $BC$, respectively, and $f_A$, $f_B$, $f_C$ single-qubit states in each individual subsystem. Since magic injection via unitary operations on any single subsystem in this setting mirrors the bipartite case discussed earlier, we first consider unitaries of the form $U=U_A\otimes U_B$ [Fig.~\ref{fig:SM2}(b)]. In the limit where $|A|\gg 1$ and $|B|\gg 1$, we find that the average $\overline{Y^{\rm lin}}$ after applying $U$ is given by
\begin{eqnarray}
\overline{Y^{\rm lin}}&=& \mathbb{E}_{U_A}\mathbb{E}_{U_B} Y^{\rm lin}\left(U_AU_B|\psi\rangle \langle \psi |U_A^\dagger U_B^\dagger\right)  \nonumber \\
&=& 1-4\cdot {2}^{-|A|-|B|-g-{{b}_{AC}}-{{b}_{BC}}}\left[1+3\cdot {2}^{-2{{b}_{AB}}-g}\right],
\label{eq:tripartite_UaUb}
\end{eqnarray}
where subleading corrections are of order $O(2^{-|A|}, 2^{-|B|})$. A detailed derivation of Eq.~(\ref{eq:tripartite_UaUb}) is given in Sec.~\ref{sec:1c}.
Eq.~(\ref{eq:tripartite_UaUb}) can be understood as follows. The average value $1-\overline{Y^{\rm lin}}$ is once again exponentially suppressed in both the total size of the subregion where the unitary is applied: $|A|+|B|$, and the amount of entanglement between subregion $AB$ and its complement: $g+b_{AC}+b_{BC}$. Hence the physical interpretation of the prefactor in Eq.~(\ref{eq:tripartite_UaUb}) is identical to that of Eq.~(\ref{eq:Y_bi}) in the main text. The discrepancy between these two cases lies in the additional factor $3\cdot {2}^{-2{{b}_{AB}}-g}$, which becomes negligible once $A$ and $B$ are sufficiently entangled. This is precisely what we have seen in Eq.~(\ref{eq:UaUb}) of the main text: the effect of a factorized unitary $U_A\otimes U_B$ becomes essentially indistinguishable from that of a global unitary from the standpoint of magic, once $A$ and $B$ are sufficiently entangled. Thus, Eq.~(\ref{eq:tripartite_UaUb}) encapsulates the combined effects of Eqs.~(\ref{eq:Y_bi}) and ~(\ref{eq:UaUb}) of the main text in the tripartite setting. In Sec.~\ref{sec:3}, we further consider a factorized unitary acting on all three subregions: $U=U_A\otimes U_B \otimes U_C$, where we confirm that in this case the 
effect of a factorized Haar random unitary on magic injection again becomes indistinguishable from that of a global Haar random unitary, provided that the subregions are sufficiently entangled with one another. 

\begin{figure}[h]
    \centering
    \includegraphics[width=0.6\linewidth]{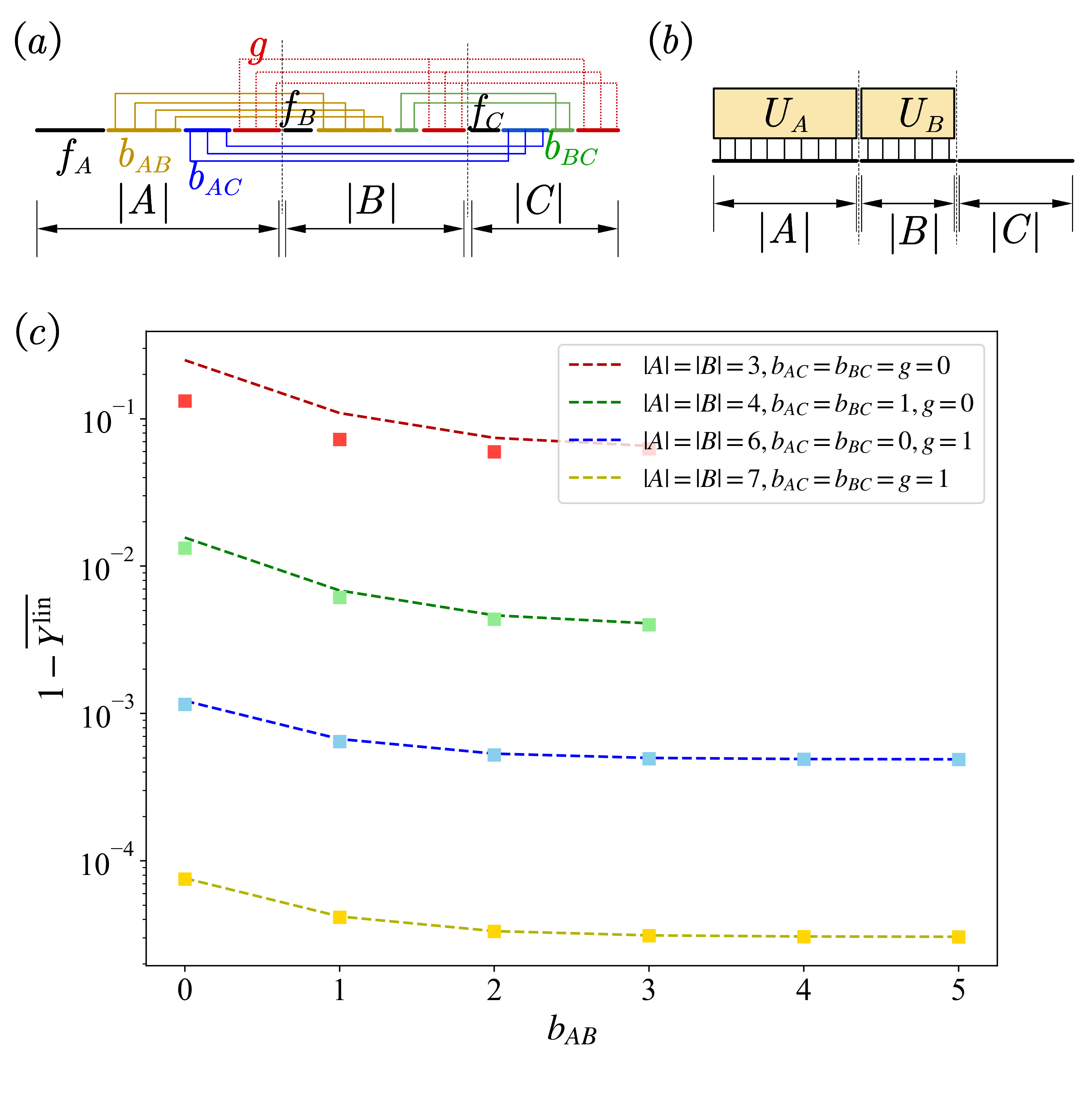}
    \caption{(a) Schematic of the initial stabilizer state used for tripartite systems. Notice that any tripartite stabilizer state can be brought into this form via local Clifford unitaries acting on subsystems $A$, $B$ and $C$. (b) Schematic of the setup for \(U_A\otimes U_B\) acting on tripartite systems. (c) Numerical results for \(U_A\otimes U_B\) acting on tripartite systems. The dashed lines represent the quantity shown in Eq.~(\ref{eq:tripartite_UaUb}) and the squares represent the exact result Eq.~(\ref{eq:SMp_23}).}
    \label{fig:SM2}
\end{figure}


\section{Analytical method}
In this section, we present the analytical method used to derive the main results quoted in the main text. In Section~\ref{sec:1a}, we provide a brief introduction to a diagrammatic approach for computing polynomials of matrix elements of unitary matrices averaged over the Haar ensemble. Section~\ref{sec:1b} introduces a key quantity used throughout our analysis and outlines approximations made to simplify both the calculation and the final results. Finally, in Section~\ref{sec:1c}, we detail the step-by-step derivation of the results presented in the main text. 
Exact results of all setups considered in the main text (i.e. without any approximation such as $|A|\gg 1$) can also be obtained using an alternative analytical method, as provided in Appendix~\ref{sec:5}. This alternative approach, while more compatible with symbolic computations, offers less physical intuition.

\subsection{A diagrammatic approach}\label{sec:1a}
In this subsection, we outline a diagrammatic method for systematically computing expectation values of polynomials in the matrix elements of unitary matrices under the Haar ensemble. Our presentation closely follows the original work of Brouwer (1996)~\cite{Brouwer_1996}. Consider a polynomial function of the form
\begin{equation}
f(U)=U_{a_1 b_1}\dots U_{a_n b_n}U^{*}_{\alpha_1 \beta_2}\dots U^{*}_{\alpha_m \beta_m}. 
\end{equation}
where $U_{ij}$ and $U_{kl}^*$ denote matrix elements of a unitary matrix and its complex conjugate, respectively. The Haar measure expectation value of such functions, denoted by $\overline{f(U)} := \mathbb{E}_{\rm Haar} [f(U)]$, is given by
\begin{equation}
    \overline{f(U)}=\delta_{nm}\sum_{\sigma,\pi}V_{\sigma,\pi}\prod\limits_{j=1}^{n}\delta_{a_j\alpha_{\sigma(j)}}\delta_{b_j\beta_{\pi(j)}}
\label{eq:SM1}\end{equation}
where the summation is over all permutations \(P\) and \(P'\) of the integers \((1,\dots,n)\). The coefficients \(V_{\sigma,\pi}\), also known as the Weingarten function, depend only on the cycle structure of the permutation \(\sigma^{-1}\pi\) \cite{diagram}. In other words, it only depends on the lengths \(c_1,\dots,c_k\) of the cycles in the factorization of \(\sigma^{-1}\pi\). So below we write \(V_{c_1,\dots,c_k}\) instead of \(V_{\sigma,\pi}\). 

The diagrams consist of the building blocks shown in Fig.~\ref{fig:SM3}. The matrix elements \(U_{ab}\) or \(U^{*}_{\alpha\beta}\) are represented by thick dotted lines. The first index (\(a\) or \(\alpha\)) is represented by a black dot, the second index (\(b\) or \(\beta\)) is represented by a white dot. 
Matrix element \(A_{ij}\) of a fixed matrix $A$ (i.e., a matrix that is not part of the average) is represented by a directed thick solid line, pointing from the first to the second index, without dots at the endpoints. The Kronecker delta is represented by an undirected thin solid line, without dots at the endpoints. Two dots connected by a solid line indicate contractions of their corresponding matrix indices.
As an example, the functions \(f(U)=\text{Tr}(AUBU^{\dagger})\) and \(g(U)=\text{Tr}(AUBUCU^{\dagger}DU^{\dagger})\) are represented in Fig.~\ref{fig:SM4}.

\begin{figure}[!tbh]
    \centering
    \includegraphics[width=0.25\linewidth]{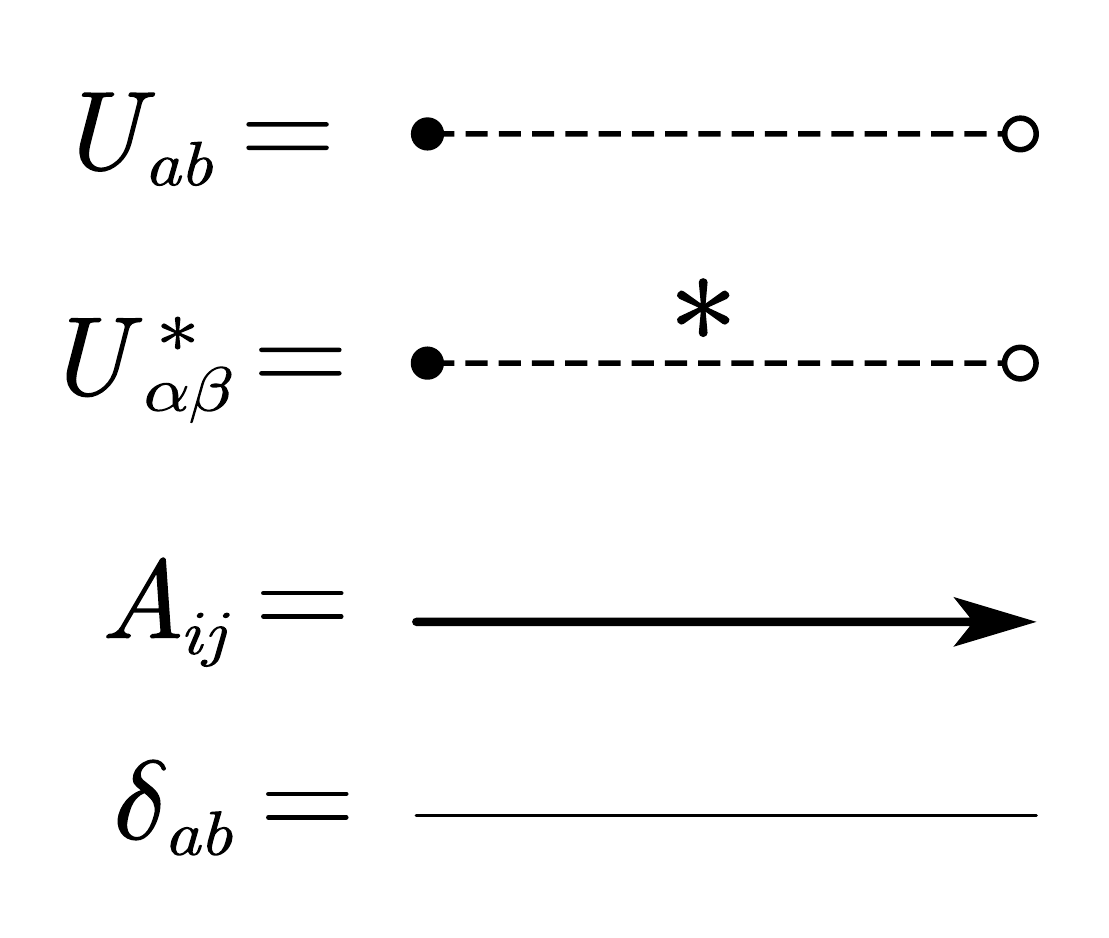}
    \caption{Diagrammatic representations for the unitary matrix \(U\) and \(U^*\), the fixed matrix \(A\) and the Kronecker delta.}
    \label{fig:SM3}
\end{figure}

\begin{figure}[!tbh]
    \centering
    \includegraphics[width=0.8\linewidth]{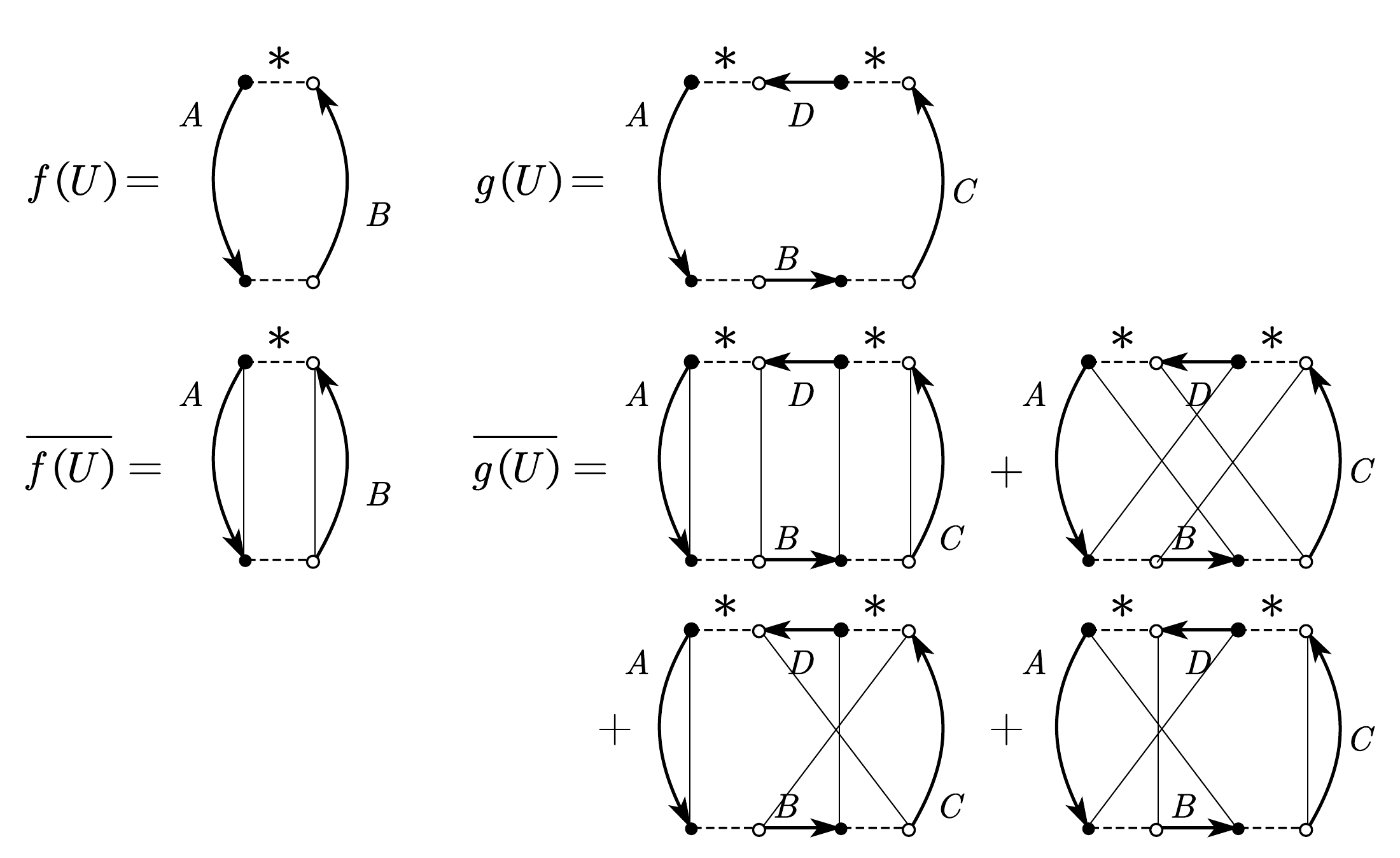}
    \caption{Top: diagrams representing \(f(U)=\text{Tr}(AUBU^{\dagger})\) and \(g(U)=\text{Tr}(AUBUCU^{\dagger}DU^{\dagger})\), respectively. Bottom: diagrams that keep track of all terms after taking the average over Haar random $U$. The average of $\overline{f(U)}$ contains only one term, whereas $\overline{g(U)}$ contains four terms.}
    \label{fig:SM4}
\end{figure}

Since the average over Haar random ensemble leads to a summation of permutations, we need to keep track of each term in the summation in Eq.~(\ref{eq:SM1}) by keeping track of $P$ and $P'$ in the diagram.
Each $P$ represents one way of pairing the first index of $U$ with that of $U^*$, and $P'$ pairs the second index of $U$ with that of $U^*$. This pairing is represented by adding additional thin lines connecting the endpoints of $U$ and that of $U^*$ in the diagram, with black dots connected to black dots, and white dots connected to white dots, as depicted in Fig.~\ref{fig:SM4}.
To associate each diagram with a specific value, we follow the following rules: 

(i) Cycle structure: A closed circuit consisting of alternating dotted or thin lines correspond to a cycle in \( P^{-1} P' \). The length \( c_k \) of it is half the number of dotted lines in the circuit. We call the circuit a \( U \)-cycle of length \( c_k \).

(ii) Matrix trace: A closed circuit consisting of alternating thick or thin lines is called a \(T\)-cycle. A \( T \)-cycle containing matrices \( A^{(1)}, A^{(2)}, \dots, A^{(k)} \) (ordered when going in one direction of the closed circuit) corresponds to \( \text{Tr}A^{(1)} A^{(2)} \dots A^{(k)} \) (recall that thin lines represent $\delta_{ab}$ that locks the two indices they connect). If a thick line corresponding to matrix \( A \) is traversed in the opposite direction, the matrix is replaced by its transpose \( A^{T} \).

To illustrate this procedure, we consider the averages of the functions \( \overline{f(U)} = \overline{\text{Tr}(AUBU^{\dagger}) }\) and \( \overline{g(U)} = \overline{\text{Tr}(AUBUCU^{\dagger}DU^{\dagger} )}\). By connecting the dots with thin lines, we obtain the diagrams shown in Fig.~\ref{fig:SM4}. For \( f \), there is only one diagram, which consists of a single \( U \)-cycle of length 1 (with weight \( V_1 \)) and two \( T \)-cycles, which generate \( \text{Tr}A \) and \( \text{Tr}B \). Thus, the expected value is:
\begin{equation}
    \overline{f(U)}=V_1 \text{Tr}A \ \text{Tr}B.
\label{eq:SM2}\end{equation}
As for \(g\), there are four diagrams that contribute. The first diagram contains two \(U\)-cycles of length \(1\), and three \(T\)-cycles. Its contribution is \(V_{1,1}\text{Tr}A\ \text{Tr}(BD) \ \text{Tr}C\). The second diagram contains two \(U\)-cycles of length \(1\) and a single \(T\)-cycle. Its contribution is \(V_{1,1}\text{Tr}(ABCD)\). The third and fourth diagram each contain a single \(U\)-cycle of length \(2\) and two \(T\)-cycles. Their contributions are \(V_2\text{Tr}A\ \text{Tr}(BDC)\) and \(V_2\text{Tr}(ADB)\ \text{Tr}C\). So we obtain:
\begin{equation}
    \overline{g(U)}=V_{1,1}\left[\text{Tr}(ABCD)+\text{Tr}A\ \text{Tr}(BD) \ \text{Tr}C\right]+V_2\left[\text{Tr}A\ \text{Tr}(BDC)+\text{Tr}(ADB) \ \text{Tr}C\right].
\label{eq:SM3}\end{equation}

\subsection{A useful quantity and its approximations}\label{sec:1b}
 In this subsection, we define a quantity that will be used intensively in evaluating $\overline{Y^{\rm lin}}$. We will calculate the value of this quantity in its most general form, to leading order in $2^{|A|+E}$, where $E$ is the entanglement entropy. 
 Let \(\mathbb{S}_{A}\) denote an \(|A|\)-qubit stabilizer group whose size is \(2^{f_A}\) with $f_A < |A|$; then let \(u_i,i=1,2,3,4 \) denote four Pauli strings supported on $A$ (where some of them can be identical) that commute with every element in \(\mathbb{S}_{A}\) (i.e. \(u_i\) is in the normalizer group of \(\mathbb{S}_{A}\)). Notice that we allow $u_i = I_A$, but otherwise $u_i \notin \mathbb{S}_A$.
 Since we assume $f_A<|A|$, one can always find a nontrivial $u_i$. And because \(u_i\mathbb{S}_A=u_ig\mathbb{S}_A,\forall g \in \mathbb{S}_A\), we can choose one specific \(u_i\). Then we define the following function:
\begin{equation}
    {{\alpha }_{{{P}_{A}}}}({{u}_{1}},{{u}_{2}},{{u}_{3}},{{u}_{4}}):= \overline{\prod\limits_{i=1}^{4}{\text{Tr}({{P}_{A}}{{U}_{A}}{{u}_{i}} \mathbb{S}_A {U}_A^\dagger)}},
\label{eq:SM4}\end{equation}
where $P_A$ is a Pauli string supported on subsystem $A$. In the above definition, we also abuse notation by using $\mathbb{S}_A$ to represent an equal weight summation of all elements in $\mathbb{S}_A$: $\sum_{g\in \mathbb{S}_A} g$.
The connection between the above quantity and $\overline{Y^{\rm lin}}$ will be explained in detail in the next subsection. Basically it naturally arises when computing ${\rm Tr}(P U_A \rho U_A^\dagger)^4$, with $\rho$ represented as a projector onto the stabilizer subspace.

Eq.~(\ref{eq:SM4}) can be represented using the diagrammatic method introduced in Sec.~\ref{sec:1a} as follows:
\begin{figure}[!tbh]
    \centering
    \includegraphics[width=\linewidth]{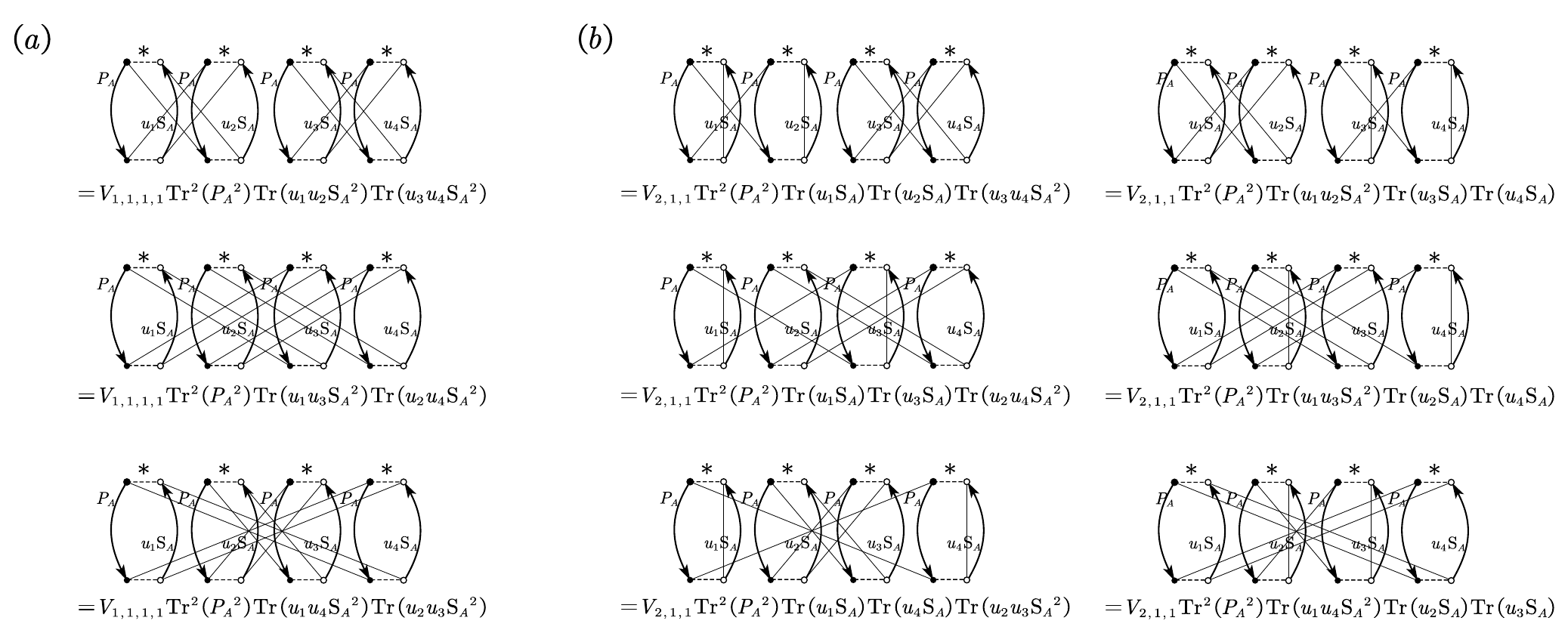}
    \caption{(a) Leading order and the corresponding diagrams for $\alpha_{P_A}(u_1,u_2,u_3,u_4)$. There are \({4!}^{2}=576\) diagrams in total. If \(P_A\ne I_A\), then $\text{Tr}(P_A)=\text{Tr}(P_A^{3})=0$, and there are 216 diagrams with non-zero contributions. (b) Subleading order and the corresponding diagrams for $\alpha_{P_A}(u_1,u_2,u_3,u_4)$.}
    \label{fig:SM5}
\end{figure}

First of all, the factors \(V_{c_1, c_2, \dots}\) can be readily obtained using known results of Weingarten functions. Let $N:=2^{|A|}$, we have:
\begin{equation}
\begin{aligned}
  & V_{1,1,1,1} =\frac{N^4 - 8N^2 + 6}{N^2(N^2-1)(N^2-4)(N^2-9)} = 2^{-4|A|} + 6 \cdot 2^{-6|A|} + O(2^{-8|A|}), \\
 & V_{2,1,1} = \frac{-N^3 + 4N}{N^2(N^2-1)(N^2-4)(N^2-9)} = -2^{-5|A|} + 18 \cdot 2^{-7|A|} + O(2^{-9|A|}), \\
 & V_{2,2} = \frac{N^2 + 6}{N^2(N^2-1)(N^2-4)(N^2-9)} = 2^{-6|A|} + 20 \cdot 2^{-8|A|} + O(2^{-10|A|}), \\
 & V_{3,1} = \frac{2N^2 - 6}{N^2(N^2-1)(N^2-4)(N^2-9)} = 2 \cdot 2^{-6|A|} + 8 \cdot 2^{-8|A|} + O(2^{-10|A|}), \\
 & V_{4} = \frac{-5N}{N^2(N^2-1)(N^2-4)(N^2-9)} = -5 \cdot 2^{-7|A|} + O(2^{-9|A|}).
\end{aligned}    
\label{eq:SM5}
\end{equation}

Careful inspections of all diagrams reveals that factors involving $P_A$ are of the form \(\text{Tr}^4 (P_A)\), \(\text{Tr}(P_A^3)\text{Tr}({P_A})\), \(\text{Tr}(P_A^2)\text{Tr}^2({P_A})\), \(\text{Tr}^2(P_A^2)\), and \(\text{Tr}(P_A^4)\). Notice that $P_A^2 = P_A^4= I_A$, $P_A^3 = P_A$. Also, if $P_A \neq I_A$, ${\rm Tr}(P_A)=0$. Thus, in this case, the only non-zero contributions come from terms involving \(\text{Tr}^2({P_A}^2) = 2^{2|A|}\) and \(\text{Tr}({P_A}^4) = 2^{|A|}\). Similarly, terms involving $u_i$ and $\mathbb{S}_A$ are of the form \(\text{Tr} (u_1 u_2 u_3 u_4 {\mathbb{S}}_A^4)\), \(\text{Tr}(u_1 u_2 u_3 {\mathbb{S}}_A^3)\text{Tr}(u_4 {\mathbb{S}}_A)\), \(\text{Tr}(u_1 u_2 {\mathbb{S}}_A^2)\text{Tr}(u_3 u_4 {\mathbb{S}}_A^2)\), \(\text{Tr}(u_1 u_2 {\mathbb{S}}_A^2)\text{Tr}(u_3 {\mathbb{S}}_A)\text{Tr}(u_4 {\mathbb{S}}_A)\), and \(\text{Tr}(u_1 {\mathbb{S}}_A)\text{Tr}(u_2 {\mathbb{S}}_A)\text{Tr}(u_3 {\mathbb{S}}_A)\text{Tr}(u_4 {\mathbb{S}}_A)\) (and also permutations of the $u_i$'s). It is also useful to remember ${\rm Tr}(\mathbb{S}_A)=2^{|A|}$, and ${\rm Tr}(\mathbb{S}_A^2)=2^{|A|+f_A}$.

When \(P_A = I_A\), $\alpha_{I_A}$ trivially becomes $\prod_{i=1}^4 {\rm Tr}(u_i \mathbb{S}_A)$, which vanishes unless $u_i=I_A$, $\forall i$. We thus have:
\begin{equation}
    {{\alpha }_{{{I}_{A}}}}({{u}_{1}},{{u}_{2}},{{u}_{3}},{{u}_{4}}) = 2^{4|A|} \delta_1(u_1) \delta_1(u_2) \delta_1(u_3) \delta_1(u_4),
\label{eq:SM6}
\end{equation}
where $\delta_1(u_i)=1$ iff $u_i=I_A$ and zero otherwise.

When \(P_A \neq I_A\), by carefully examining the structures of all diagrams, we find that the leading order contributions to $\alpha_{P_A}$ come from the diagrams shown in Fig.~\ref{fig:SM5}. 
\begin{equation}
\begin{aligned}
    &{{\alpha }_{{{P}_{A}}}}({{u}_{1}},{{u}_{2}},{{u}_{3}},{{u}_{4}}) \\
    &= 2^{2f_A} \left( \underbrace{{{\delta }_{2}}({{u}_{1}},{{u}_{2}}){{\delta }_{2}}({{u}_{3}},{{u}_{4}}) + \text{permutations}}_{3 \text{ terms}} \right)  \\
    &\quad -2^{f_A} \left( \underbrace{{{\delta }_{2}}({{u}_{1}},{{u}_{2}}){{\delta }_{1}}({{u}_{3}}){{\delta }_{1}}({{u}_{4}}) + \text{permutations}}_{\text{6 terms}} \right) \\
    & \quad+O \left( 2^{3f_A - 2|A|} \delta_4(u_1, u_2, u_3, u_4) \right),
\end{aligned}
\label{eq:SM7}
\end{equation}
where \(\delta_2(u_1, u_2) = 1\) iff \(u_1 = u_2\), and zero otherwise; \(\delta_4(u_1, u_2, u_3, u_4) = \frac{\text{Tr}(u_1 u_2 u_3 u_4)}{2^{|A|}}\). In the above equation, the first line comes from terms like $V_{1,1,1,1} {\rm Tr}^2(P_A^2) {\rm Tr}(u_1 u_2 \mathbb{S}_A^2) {\rm Tr}(u_3 u_4 \mathbb{S}_A^2)$, and the 2nd line comes from $V_{2,1,1} {\rm Tr}^2(P_A^2) {\rm Tr}(u_1 u_2 \mathbb{S}_A^2) {\rm Tr}(u_3 \mathbb{S}_A) {\rm Tr}(u_4\mathbb{S}_A)$.

\subsection{Derivation of the key results presented in the main text}\label{sec:1c}

In this subsection, we provide a detailed derivation of Eqs.~(\ref{eq:Y_bi}),~(\ref{eq:UaUb}) in the main text, as well as Eq.~(\ref{eq:tripartite_UaUb}) in Sec.~\ref{sec:tri}. Since what we essentially need to compute is $1-\overline{Y^{\rm lin}}$, we will focus on this part below.

\subsubsection{Eq.~(\ref{eq:Y_bi})}

We start with Eq.~(\ref{eq:Y_bi}). The key to our derivation is to represent the initial state as a projector onto the stabilizer subspace, and a decomposition of the stabilizer group of the full system in terms of the cosets of stabilizer subgroups of subsystems $A$ and $B$. Let $\mathbb{S}_{AB}$ denote the stabilizer group of the initial state. Then initial density matrix can be written as a projector
\begin{equation}
\rho_0 = \frac{1}{2^N} \sum_{g \in \mathbb{S}_{AB}} g.
\end{equation}
Let $\mathbb{S}_A$ and $\mathbb{S}_B$ denote the stabilizer subgroups of subsystem $A$ and $B$, respectively. Let $E$ denote the entanglement entropy between subsystem $A$ and $B$. We have $E=|A|-{\rm dim}(\mathbb{S}_A) = |B|-{\rm dim}(\mathbb{S}_B)$.
$\mathbb{S}_{AB}$ can then be decomposed in terms of the cosets of $\mathbb{S}_A$ and $\mathbb{S}_B$ as follows:
\begin{equation}
\mathbb{S}_{AB} = \underbrace{\mathbb{S}_A \otimes \mathbb{S}_B \ \cup \ a_1 \mathbb{S}_A \otimes b_1 \mathbb{S}_B \ \cup \ a_2 \mathbb{S}_A \otimes b_2 \mathbb{S}_B \ \cup \ldots}_{4^{E}\ \rm{terms}}
\end{equation}
where $a_i \notin \mathbb{S}_A$ and $b_i \notin \mathbb{S}_B$ are the logical operators of subsystems $A$ and $B$, respectively. To ensure the proper commutation relation, the logical operators must be paired up into $(a_i, b_i)$ such that: (i) if $[a_i, a_j]=0$, then $[b_i, b_j]=0$; or (ii) if $\{a_i, a_j\}=0$, then $\{b_i, b_j\}=0$. 
While it is not hard to convince oneself that the above decomposition is correct for pure states with ${\rm dim}(\mathbb{S}_{AB})=N$, let us check that the dimensions work out. First of all, since the entanglement between $A$ and $B$ is $E$, there are precisely $E$ logical qubits encoded in both subsystems, with $4^E$ logical operators in each subsystem, and hence $4^E$ terms in the decomposition. Secondly, we can count the total number of group elements in $\mathbb{S}_{AB}$ from the right-hand side of the decomposition. We have $4^E \times 2^{{\rm dim}(\mathbb{S}_A)} \times 2^{{\rm dim}(\mathbb{S}_B)}=2^{|A|+|B|}=2^N$, as it should. We will give a concrete example of such a decomposition in Appendix~\ref{sec:2}.

We can now express the initial state using the above coset decomposition of $\mathbb{S}_{AB}$:
\begin{equation}
    {{\rho }_{0}}=\frac{1}{{2}^N}{\mathbb{S}_{AB}}=  \frac{1}{2^N} \left( \underbrace{\mathbb{S}_A \otimes \mathbb{S}_B \ + \ a_1 \mathbb{S}_A \otimes b_1 \mathbb{S}_B \ + \ a_2 \mathbb{S}_A \otimes b_2 \mathbb{S}_B \ + \ldots}_{4^{E}\ \rm{terms}} \right)
\label{eq:SM9}\end{equation}

Armed with the above representation of $\rho_0$, we can cast the average value $\overline{Y^{\rm lin}}$ as:
\begin{equation}
\begin{aligned}&
    1-\overline{Y^{\rm lin}}=\frac{1}{{{2}^{5N}}}\sum\limits_{P \in P_N}{\overline{{{\text{Tr}}^{4}}(P{{U}_{A}}{\mathbb{S}_{AB}}{{U}_{A}}^{\dagger })}}\\
    &=\frac{1}{{{2}^{5N}}}\sum_{P\in P_N}{\sum\limits_{k=0}^{{{2}^{2E}-1}}{{{\alpha }_{{{P}_{A}}}}({{a}_{k}},{{a}_{k}},{{a}_{k}},{{a}_{k}})\ {{\text{Tr}}^{4}}({{P}_{B}}{{b}_{k}}{{\mathbb{S}}_{B}})}}.
\end{aligned}
\label{eq:SM10}\end{equation}
In going from the first to the second line, notice that the Pauli string $P=P_A \otimes P_B$, and that the unitary $U_A$ is supported on subsystem $A$ and hence leaves $P_B$ invariant. Moreover, since different cosets $b_j \mathbb{S}_B$ do not overlap, for a given $P$ (and $P_B$), there is at most one coset such that ${\text{Tr}^{4}({{P}_{B}}{{b}_{k}}{{\mathbb{S}}_{B}})}$ does not vanish. The trace involving $P_A$ is now exactly of the form in Eq.~(\ref{eq:SM4}), with $u_1=u_2=u_3=u_4=a_k$ for some $k$. We can now directly apply our general results in Eqs.~(\ref{eq:SM6}) and~(\ref{eq:SM7}).

Consider first a fixed Pauli string $P$ and $k$, which corresponds to one term of the summation in $\overline{Y^{\rm lin}}$. We have the following result:

\begin{eqnarray}
&&{{{\alpha }_{{{P}_{A}}}}({{a}_{k}},{{a}_{k}},{{a}_{k}},{{a}_{k}})\ \text{Tr}^{4}({{P}_{B}}{{b}_{k}}{{\mathbb{S}}_{B}})} \nonumber \\
&=& \delta(P_B, b_k \mathbb{S}_B) \times
\begin{cases}
   3\cdot {{2}^{2{|A|}+4{|B|}-2E}} \left[ 1 - 2\cdot {{2}^{E-{|A|}}} + O({{2}^{-E-{|A|}}}) \right], & \quad {{P}_{A}}\ne {{I}_{A}}, \quad k=0  \\
   3\cdot {{2}^{2{|A|}+4{|B|}-2E}} \left[ 1 + O({{2}^{-E-{|A|}}}) \right], & \quad {{P}_{A}}\ne {{I}_{A}}, \quad k \neq 0 \\
   {{2}^{4{|A|}+4{|B|}}}, & \quad {{P}_{A}}={{I}_{A}}, \quad k=0  \\
   0, & \quad {{P}_{A}}={{I}_{A}}, \quad k \neq 0  \\
\end{cases}
\label{eq:SM11}
\end{eqnarray}
In the above equation, $k=0$ corresponds to the first term in the coset decomposition, with $a_0=b_0= I$; $\delta(P_B, b_k \mathbb{S}_B)=1$ iff $P_B \in b_k \mathbb{S}_B$, and zero otherwise. Summing over $P\in P_N$ and $k$, we have:
\begin{eqnarray}
&&\sum\limits_{P\in P_N}{\sum\limits_{k=0}^{{{2}^{2E}-1}}{{{\alpha }_{{{P}_{A}}}}({{a}_{k}},{{a}_{k}},{{a}_{k}},{{a}_{k}})\ {{\text{Tr}}^{4}}({{P}_{B}}{{b}_{k}}{{\mathbb{S}}_{B}})}} \nonumber\\
=&&\underbrace{{2}^{|B|-E}\cdot(2^{2|A|}-1)\cdot3\cdot {{2}^{2{|A|}+4{|B|}-2E}} \left[ 1 - 2\cdot {{2}^{E-{|A|}}}  O({{2}^{-E-{|A|}}}) \right]}_{P_A\ne I_A,\quad k=0} \nonumber\\
+&&\underbrace{{2}^{|B|-E}\cdot(2^{2E}-1)\cdot(2^{2|A|}-1)\cdot3\cdot {{2}^{2{|A|}+4{|B|}-2E}} \left[ 1 + O({{2}^{-E-{|A|}}}) \right]}_{P_A\ne I_A,\quad k\ne 0} \nonumber\\
+&&\underbrace{{2}^{|B|-E}\cdot{2}^{4|A|+4|B|}}_{P_A=I_A,\quad k= 0} \nonumber\\
=&&4\cdot{2}^{4|A|+5|B|-E}\left[1+O({2}^{-|A|-E})\right].
\label{eq:SM11_1}
\end{eqnarray}
In the above summation, the prefactor $(2^{2|A|}-1)$ counts the total number of $P_A \neq I_A$; $(2^{|B|-E})$ comes from the size of each coset $|b_k \mathbb{S}_B|$. The final result is precisely Eq.~(\ref{eq:Y_bi}) quoted in the main text after dividing by $2^{5N}$.

The qualitative features of Eq.~(\ref{eq:Y_bi}), particularly the factor $2^{-|A|-E}$, can also be understood---since the unitary acts only on subsystem $A$, stabilizers supported entirely on $B$ remain unchanged under $U_A$ and contribute to $1-Y^{\rm lin}(\rho)$. A simple counting shows that the number of surviving stabilizers is precisely $2^{|B|-E}$, leading to the factor $2^{-|A|-E}$ in Eq.~(\ref{eq:Y_bi}). The prefactor 4 in front of $2^{-|A|-E}$ is, in some sense, non-universal. To see this, let us define the generalized linear stabilizer entropy $Y_{\alpha}^{\rm lin}(\rho) := 1-2^{-N}\sum_{P\in P_N}{\rm tr}(P \rho)^{2\alpha}$. Eq.~(\ref{eq:Y_bi}) thus corresponds to the specific case $\alpha=2$. Similarly, $1-Y_{\alpha}^{\rm lin}(\rho)$ is related to the $\alpha$-th stabilizer Renyi entropy: $M_\alpha(\rho)=\frac{1}{1-\alpha}{\rm log} (1-Y_{\alpha}^{\rm lin}(\rho))$. In the limit $\alpha \rightarrow \infty$, only the surviving stabilizers supported on $B$ contribute, yielding $\overline{Y_{\infty}^{\rm lin}} = 1-2^{-|A|-E}$. For arbitrary $\alpha$, we have $1-\overline{Y_{\alpha}^{\rm lin}}= [1+c_{\alpha}(|A|,E)] 2^{-|A|-E}$, with the one coming from the surviving stabilizers, and $c_{\alpha}(|A|,E)$ coming from contributions of all other Pauli strings, which in general depends on both $|A|$ and $E$~\footnote{We thank Xhek Turkeshi for pointing this out to us.}. Specifically, Eq.~(\ref{eq:Y_bi}) implies that for $\alpha=2$, we have $c_2=3$.

\subsubsection{Eq.~(\ref{eq:UaUb})}

We provide detailed calculations for Eq.~(\ref{eq:UaUb}) in the main text. Since the unitary acting on $\rho_0$ takes a factorized form $U=U_A \otimes U_B$, the average $\overline{Y^{\rm lin}}$ similarly factorizes into
\begin{equation}
\begin{aligned}
1-\overline{Y^{\rm lin}}=\frac{1}{{{2}^{5N}}}\sum\limits_{P\in P_N}{\sum\limits_{{{k}_{1}},{{k}_{2}},{{k}_{3}},{{k}_{4}}}{{{\alpha }_{{{P}_{A}}}}({{a}_{{{k}_{1}}}},{{a}_{{{k}_{2}}}},{{a}_{{{k}_{3}}}},{{a}_{{{k}_{4}}}}){{\beta }_{{{P}_{B}}}}({{b}_{{{k}_{1}}}},{{b}_{{{k}_{2}}}},{{b}_{{{k}_{3}}}},{{b}_{{{k}_{4}}}})}},
\end{aligned}
\label{eq:SM12}\end{equation}
where $\beta_{P_B} \left( b_{k_1}, b_{k_2}, b_{k_3}, b_{k_4} \right)$ is defined analogously to Eq.~(\ref{eq:SM4}). 

For convenience, we define the following functions: \[f_1:=\underbrace{{{\delta }_{2}}({{a}_{k_1}},{{a}_{k_2}}){{\delta }_{2}}({{a}_{k_3}},{{a}_{k_4}})+\text{permutations}}_{3\text{ terms}},\]
\[f_2:= \underbrace{{{\delta }_{2}}({{a}_{k_1}},{{a}_{k_2}}){{\delta }_{1}}({{a}_{k_3}}){{\delta }_{1}}({{a}_{k_4}})+\text{permutations}}_{\text{6 terms}},\] \[f_3:= \delta_4({a}_{k_1},{a}_{k_2},{a}_{k_3},{a}_{k_4}),\]
where the definition of $\delta_2$ and $\delta_4$ was given below Eq.~(\ref{eq:SM7}). The above functions $f_1, f_2,$ and $f_3$ correspond to the three leading order contributions in $\alpha_{P_A\neq I_A}$ as given in Eq.~(\ref{eq:SM7}). We suppress the indices $k_1,\ldots, k_4$ in the definition of $f_i$'s to simplify notation.
We then have the following summations:
\begin{equation*}
    \sum\limits_{{{k}_{1}},{{k}_{2}},{{k}_{3}},{{k}_{4}}}{{{f}_{1}}}=3\cdot 2^{4E},\quad \sum\limits_{{{k}_{1}},{{k}_{2}},{{k}_{3}},{{k}_{4}}}{{{f}_{2}}}=6\cdot 2^{2E},\quad \sum\limits_{{{k}_{1}},{{k}_{2}},{{k}_{3}},{{k}_{4}}}{{{f}_{3}}}\leq2^{6E}.
\end{equation*}

Now, Eq.~(\ref{eq:SM12}) can be readily evaluated by summing over contributions from (i) $P_A \neq I_A$, $P_B \neq I_B$; (ii) $P_A \neq I_A$, $P_B=I_B$; (iii) $P_A=I_A$, $P_B \neq I_B$; (iv) $P_A=I_A$, $P_B=I_B$, using Eq.~(\ref{eq:SM7}). Let us consider contributions from each case.

(i) \(P_A\ne I_A\) and \(P_B\ne I_B\), we have:

\begin{equation}
\begin{aligned}
    &\sum\limits_{{{k}_{1}},{{k}_{2}},{{k}_{3}},{{k}_{4}}}{{{\alpha }_{{{P}_{A}}}}({{a}_{{{k}_{1}}}},{{a}_{{{k}_{2}}}},{{a}_{{{k}_{3}}}},{{a}_{{{k}_{4}}}}){{\beta }_{{{P}_{B}}}}({{b}_{{{k}_{1}}}},{{b}_{{{k}_{2}}}},{{b}_{{{k}_{3}}}},{{b}_{{{k}_{4}}}})}\\
    &={{2}^{2{|A|}+2{|B|}-4E}}\sum\limits_{{{k}_{1}},{{k}_{2}},{{k}_{3}},{{k}_{4}}}{{{f}_{1}}{{f}_{1}}}-{{2}^{{|A|}+2{|B|}-3E}}\sum\limits_{{{k}_{1}},{{k}_{2}},{{k}_{3}},{{k}_{4}}}{{{f}_{1}}{{f}_{2}}}+\\&-{{2}^{2{|A|}+{|B|}-3E}}\sum\limits_{{{k}_{1}},{{k}_{2}},{{k}_{3}},{{k}_{4}}}{{{f}_{1}}{{f}_{2}}}+{{2}^{{|A|}+{|B|}-2E}}\sum\limits_{{{k}_{1}},{{k}_{2}},{{k}_{3}},{{k}_{4}}}{{{f}_{2}}{{f}_{2}}}+\\&+{{c}_{1}}\cdot {{2}^{{|A|}+2{|B|}-5E}}\sum\limits_{{{k}_{1}},{{k}_{2}},{{k}_{3}},{{k}_{4}}}{{{f}_{1}}{{f}_{3}}}+{{c}_{1}}\cdot {{2}^{2{|A|}+{|B|}-5E}}\sum\limits_{{{k}_{1}},{{k}_{2}},{{k}_{3}},{{k}_{4}}}{{{f}_{1}}{{f}_{3}}} +\\&+{{c}_{2}}\cdot {{2}^{{|A|}+{|B|}-6E}}\sum\limits_{{{k}_{1}},{{k}_{2}},{{k}_{3}},{{k}_{4}}}{{{f}_{3}}{{f}_{3}}}+{{c}_{3}}\cdot {{2}^{{|A|}+{|B|}-4E}}\sum\limits_{{{k}_{1}},{{k}_{2}},{{k}_{3}},{{k}_{4}}}{{{f}_{2}}{{f}_{3}}} 
\end{aligned}\label{eq:SM13}
\end{equation}
We now analyze the order of magnitude of each term. Since \(\sum\limits_{{{k}_{1}},{{k}_{2}},{{k}_{3}},{{k}_{4}}}f_i\cdot (...)\leq \sum\limits_{{{k}_{1}},{{k}_{2}},{{k}_{3}},{{k}_{4}}}f_i\), where \((...)\) represents any product of \(\delta_1,\delta_2\), or \(\delta_4\) (as additional $\delta$-functions further reduce the number of unconstrained summations over $k_1,\ldots,k_4$), we have
\begin{equation}
\sum_{k_1, k_2, k_3, k_4} f_i^2 = \sum_{k_1, k_2, k_3, k_4} f_i + O(2^{-E}), \quad   \sum_{k_1, k_2, k_3, k_4} f_i f_j \sim O(2^{-E}) \times {\rm min}(\sum f_i, \sum f_j).
\end{equation}
A careful inspection of Eq.~(\ref{eq:SM13}) reveals that, in the limit $|A|\gg1, |B| \gg 1$, the leading-order contribution comes from the term involving $\sum f_1^2$. We then have
\begin{equation}
\begin{aligned}
&\sum\limits_{{{k}_{1}},{{k}_{2}},{{k}_{3}},{{k}_{4}}}{{{\alpha }_{{{P}_{A}}}}({{a}_{{{k}_{1}}}},{{a}_{{{k}_{2}}}},{{a}_{{{k}_{3}}}},{{a}_{{{k}_{4}}}}){{\beta }_{{{P}_{B}}}}({{b}_{{{k}_{1}}}},{{b}_{{{k}_{2}}}},{{b}_{{{k}_{3}}}},{{b}_{{{k}_{4}}}})}\\
&=3\cdot{{2}^{2{|A|}+2{|B|}}} \left[ 1+2\cdot 2^{-2E}+O(2^{-|A|-E})+O(2^{-|B|-E}) \right].
\end{aligned}\label{eq:SM14}
\end{equation}

(ii) \(P_A= I_A\) and \(P_B\ne I_B\). Recall $\alpha_{I_A}(a_{k_1}, a_{k_2}, a_{k_3}, a_{k_4}) = 2^{4|A|} \delta_1(a_{k_1}) \delta_1(a_{k_2}) \delta_1(a_{k_3}) \delta_1(a_{k_4})$. We have:
\begin{equation}
\begin{aligned}
    &\sum\limits_{{{k}_{1}},{{k}_{2}},{{k}_{3}},{{k}_{4}}}{{{\alpha }_{{{I}_{A}}}}({{a}_{{{k}_{1}}}},{{a}_{{{k}_{2}}}},{{a}_{{{k}_{3}}}},{{a}_{{{k}_{4}}}}){{\beta }_{{{P}_{B}}}}({{b}_{{{k}_{1}}}},{{b}_{{{k}_{2}}}},{{b}_{{{k}_{3}}}},{{b}_{{{k}_{4}}}})}\\
    &=3\cdot 2^{4|A|+2|B|-2E}\left[1+O(2^{E-|B|})+O(2^{-E-|B|})\right]
\end{aligned}\label{eq:SM15}
\end{equation}

A similar result holds when \(P_A\ne I_A\) and \(P_B=I_B\), with $A$ and $B$ exchanged. 

(iii) \(P_A= I_A\) and \(P_B= I_B\). In this case the result is trivially \(2^{4|A|+4|B|}\). 

After combining all these contributions and including the multiplicity factors coming from the number of Pauli strings in each case, we obtain the final expression in Eq.~(\ref{eq:UaUb}) of the main text.

\subsubsection{Eq.~(\ref{eq:tripartite_UaUb})}
We now present the calculation details of Eq.~(\ref{eq:tripartite_UaUb}) in Sec.~\ref{sec:tri}. Let us first recall the setup we use for this case. We consider a tripartite system $A$, $B$, and $C$. Since any tripartite stabilizer state can be transformed into a set of single-qubit states, Bell states, and GHZ states via local Clifford unitaries, we consider initial state consisting of $g$ GHZ states shared among $A$, $B$, and $C$; $b_{AB}$, $b_{AC}$, and $b_{BC}$ Bell pairs shared between subsystems $AB$, $AC$, and $BC$, respectively, and $f_A$, $f_B$, $f_C$ single-qubit states in each individual subsystem, shown in Fig.~\ref{fig:SM2}(a). Let us start by explicitly writing down a set of stabilizer generators for this state. 

For the single-qubit states, we have 
\begin{equation}
\mathbb{S}_A = \langle X_1, X_2, X_3, \ldots, X_{f_A} \rangle,
\end{equation}
and similarly for $\mathbb{S}_B$ and $\mathbb{S}_C$. As the notation suggests, these are also the stabilizer subgroups on each individual subsystem. For the Bell-pair states shared between subsystem $A$ and $B$, we have
\begin{equation}
L_{AB}=\langle X_{f_A+1}X_{|A|+f_B+1}, Z_{f_A+1}Z_{|A|+f_B+1}, \dots, X_{f_A+b_{AB}}X_{|A|+f_B+b_{AB}}, Z_{f_A+b_{AB}}Z_{|A|+f_B+b_{AB}} \rangle,
\end{equation}
and similarly for $L_{BC}$ and $L_{AC}$. Finally, for the GHZ state, we define
\begin{eqnarray}
H_X &=& \langle X_{|A|-g+1}X_{|A|+|B|-g+1}X_{|A|+|B|+|C|-g}, \dots \rangle  \nonumber  \\
H_{AB} &=& \langle Z_{|A|-g+1}Z_{|A|+|B|-g+1}, \dots \rangle   \nonumber \\
H_{AC} &=& \langle Z_{|A|-g+1}Z_{|A|+|B|+|C|-g+1}, \dots \rangle   \nonumber \\
H_{BC} &=& \langle Z_{|A|+|B|-g+1}Z_{|A|+|B|+|C|-g+1}, \dots \rangle.
\end{eqnarray}
The set $H_X$, combined with any two of \(H_{AB}\), \(H_{AC}\), and \(H_{BC}\), generates the GHZ states shared among subsystems \(A\), \(B\), and \(C\).

Our approach of moving forward parallels that in the bipartite case. Namely, we decompose the stabilizer group into cosets of $\mathbb{S}_{AB}$ and $\mathbb{S}_C$, the stabilizer subgroups of subsystem $AB$ and $C$, respectively. Then we express the initial state in terms of the coset decomposition.
The stabilizer group of subsystem \({\mathbb{S}}_{AB}\) is given by
\begin{equation}
 {\mathbb{S}_{AB}} =  \langle \mathbb{S}_A,\  \mathbb{S}_B, \ L_{AB}, \ H_{AB} \rangle.
\label{eq:SM21}\end{equation}
The stabilizer group of total system then can be decomposed as
\begin{equation}
    {\mathbb{S}_{ABC}}=\underbrace{{{\mathbb{S}}_{AB}}\otimes {{\mathbb{S}}_{C}}\ \cup \ l_{AB}^{1}{{\mathbb{S}}_{AB}}\otimes l_{C}^{1}{{\mathbb{S}}_{C}}\ \cup \ ...}_{{{4}^{{{b}_{AC}}+{{b}_{BC}}+g}}\text{ terms}},
\label{eq:SM22}\end{equation}
where $l_{AB}^k$ and $l_C^k$ are logical operators on $AB$ and $C$ as before. In this case, since we are dealing with initial state with a factorized structure, it is easy to directly identify the logical operators. For example, $\{l_{AB}^k\}$ consists of Pauli operators acting on $b_{AC}+b_{BC}$ Bell pairs, as well as Pauli operators acting on $g$ GHZ states, hence $4^{b_{AC}+b_{BC}+g}$ operators in total. As before, to ensure elements across different cosets have the correct commutation relation as stabilizer group elements, we require (i) if $[l_{AB}^i, l_{AB}^j]=0$, then $[l_C^i, l_C^j]=0$; or (ii) if $\{l_{AB}^i, l_{AB}^j\}=0$, then $\{l_{C}^i, l_C^j\}=0$.

The average $\overline{Y^{\rm lin}}$ can be written as
\begin{equation}
\begin{aligned}&
    1-\overline{Y^{\rm lin}}=\frac{1}{{{2}^{5N}}}\sum\limits_{P \in P_N}{\overline{{{\text{Tr}}^{4}}(P{{U}_{AB}\otimes I_C}{\mathbb{S}_{ABC}}{U}_{AB}^{\dagger }\otimes I_C)}}\\
    &=\frac{1}{{{2}^{5N}}}\sum_{P\in P_N}{\sum\limits_{k}{{\overline{{{\text{Tr}}^{4}}(P_{AB}{{U}_{AB}}l_{AB}^k{\mathbb{S}_{AB}}{U}_{AB}^{\dagger})}}\ {{\text{Tr}}^{4}}({{P}_{C}}{{l}_{C}^k}{{\mathbb{S}}_{C}})}},
\end{aligned}
\label{eq:SM23_1}\end{equation}
where $U_{AB}=U_A \otimes U_B$.

Now recall that $U_{AB}=U_A \otimes U_B$ in the current setup. Therefore, the term ${{\overline{{{\text{Tr}}^{4}}(P_{AB}{{U}_{AB}}l_{AB}^k{\mathbb{S}_{AB}}{U}_{AB}^{\dagger})}}}$ has a similar factorization as in Eq.~(\ref{eq:SM12}). However, additional complications occur in this case, since, depending on $l_{AB}^k$, the set $l_{AB}^k \mathbb{S}_{AB}$ is further decomposed in different ways. In what follows, we use \(t_{ij}\) to represent the value of \(\overline{{{\text{Tr}}^{4}}({P}_{AB}{U}_{AB}l_{AB}^{k}{\mathbb{S}}_{AB}{U}_{AB}^{\dagger })}\) in each case.

(1) \(l_{AB}^{k}=I_{AB}\). We can decompose \({\mathbb{S}}_{AB}\) into:
\begin{equation}
    {\mathbb{S}}_{AB}=\underbrace{{\mathbb{S}}_{A}\otimes {\mathbb{S}}_{B} \ \cup \
    u_1 {\mathbb{S}}_{A} \otimes v_1 {\mathbb{S}}_{B}\ \cup \dots}_{2^{2b_{AB}+g}\text{ terms}}
\label{eq:SM23}\end{equation}
In this case, the quantity \(\overline{{{\text{Tr}}^{4}}({P}_{AB}{U}_{AB}l_{AB}^{k}{\mathbb{S}}_{AB}{U}_{AB}^{\dagger })}\) factorizes into $\alpha_{P_A}(u_{k_1}, u_{k_2}, u_{k_3}, u_{k_4})$ and $\beta_{P_B}(v_{k_1}, v_{k_2}, v_{k_3}, v_{k_4})$ in a similar way as Eq.~(\ref{eq:SM12}). We can thus directly apply our previous results Eqs.~(\ref{eq:SM7}), (\ref{eq:SM13}), (\ref{eq:SM14}), and (\ref{eq:SM15}), noticing that
\begin{equation}
{\rm dim}(\mathbb{S}_A) = f_A, \quad {\rm dim}(\mathbb{S}_B) = f_B,  \quad  \sum_{k_1, k_2, k_3, k_4} f_1 = 3\cdot 2^{4b_{AB} + 2g}.
\end{equation}

\begin{itemize}
 \item If \({{P}_{A}}\ne {{I}_{A}},{{P}_{B}}\ne {{I}_{B}}\), we have:
\begin{equation}
\begin{aligned}
    t_{11}=3\cdot {{2}^{2{{f}_{A}}+2{{f}_{B}}+4{{b}_{AB}}+2g}}\left( \begin{aligned}
  & 1+2\cdot {{2}^{-2{{b}_{AB}}-g}}+\\
  &O({{2}^{-{|A|}+{{b}_{AC}}-{{b}_{AB}}}})+O({{2}^{-{|B|}+{{b}_{BC}}-{{b}_{AB}}}}) \\ 
\end{aligned} \right)
\end{aligned}
\label{eq:SM24}\end{equation}

\item If \({{P}_{A}}= {{I}_{A}},{{P}_{B}}\ne {{I}_{B}}\), we have:
\begin{equation}
    t_{12}=3\cdot {2}^{4{|A|}+2{{f}_{B}}}\left[1+O({{2}^{b_{AB}+b_{BC}+g-|B|}})\right]
\label{eq:SM25}\end{equation}

\item If \({{P}_{A}}\ne {{I}_{A}},{{P}_{B}}= {{I}_{B}}\), we have:
\begin{equation}
    t_{13}=3\cdot {2}^{2{{f}_{A}}+4{|B|}}\left[1+O({{2}^{b_{AB}+b_{AC}+g-|A|}})\right]
\label{eq:SM26}\end{equation}

\item If \({{P}_{A}}= {{I}_{A}},{{P}_{B}}= {{I}_{B}}\), we have:
\begin{equation}
    t_{14}=2^{4|A|+4|B|}.
\label{eq:SM27}\end{equation}

\end{itemize}

(2) \(l_{AB}^{k} \otimes l_{C}^{k} \in  H_{AC}  - I\). There are \(N_2:=2^g-1\) terms, and we can decompose \(l_{AB}^{k}{{\mathbb{S}}_{AB}}\) into:

\begin{equation}
l_{AB}^{k}{{\mathbb{S}}_{AB}}=\underbrace{{\mathbb{S}}_{A}\otimes v_1 {\mathbb{S}}_{B}+
    u_2 {\mathbb{S}}_{A} \otimes {\mathbb{S}}_{B}+u_3 {\mathbb{S}}_{A} \otimes v_3 {\mathbb{S}}_{B}...}_{2^{2b_{AB}+g}\text{ terms}}.
\label{eq:SM28}\end{equation}
To see that the decomposition above is correct, first notice that in this case $l_{AB}^k$ are supported on subsystem $A$ only, so there must be a term $u_2\mathbb{S}_A\otimes \mathbb{S}_B$ where $\mathbb{S}_B$ is left untouched. Secondly, since $H_{AB} \subset \mathbb{S}_{AB}$, for $l_{AB}^{k} \otimes l_{C}^{k} \in  H_{AC}  - I$, it must share a common Pauli-$Z$ operator with an element in $H_{AB}$ acting on the same GHZ state spanning $A$, $B$, and $C$. Thus, there will also be a term $\mathbb{S}_A \otimes v_1 \mathbb{S}_B$ where $\mathbb{S}_A$ is recovered. 

We can then immediately obtain the following results depending on the operator content of the Pauli string $P$ on subsystem $A$ and $B$, respectively.

\begin{itemize}
\item If \({{P}_{A}}\ne {{I}_{A}},{{P}_{B}}\ne {{I}_{B}}\), we have:
\begin{equation}
\begin{aligned}
    t_{21}=3\cdot {{2}^{2{{f}_{A}}+2{{f}_{B}}+4{{b}_{AB}}+2g}}\left( \begin{aligned}
  & 1+2\cdot {{2}^{-2{{b}_{AB}}-g}}+\\
  &O({{2}^{-{|A|}+{{b}_{AC}}-{{b}_{AB}}}})+O({{2}^{-{|B|}+{{b}_{BC}}-{{b}_{AB}}}}) \\ 
\end{aligned} \right)
\end{aligned}
\label{eq:SM29}\end{equation}

\item If \({{P}_{A}}= {{I}_{A}},{{P}_{B}}\ne {{I}_{B}}\), we have:
\begin{equation}
    t_{22}=3\cdot {2}^{4{|A|}+2{{f}_{B}}}\left[1+O({{2}^{-b_{AB}-b_{BC}-g-|B|}})\right]
\label{eq:SM30}\end{equation}

\item If \({{P}_{A}}\ne {{I}_{A}},{{P}_{B}}= {{I}_{B}}\), we have:
\begin{equation}
    t_{23}=3\cdot {2}^{2{{f}_{A}}+4{|B|}}\left[1+O({{2}^{-b_{AB}-b_{AC}-g-|A|}})\right]
\label{eq:SM31}\end{equation}

\item If \({{P}_{A}}= {{I}_{A}},{{P}_{B}}= {{I}_{B}}\), we have:
\begin{equation}
    t_{24}=0,
\label{eq:SM32}\end{equation}
since $I_{AB} \notin l_{AB}^k S_{AB}$.

\end{itemize}

(3) \(l_{AB}^{k} \otimes l_{C}^{k} \in \left\langle L_{BC},\ H_{BC}\right\rangle-  H_{BC} \). There are \(N_3:=2^{2b_{BC}+g}-2^g\) terms. The reason we must exclude $H_{BC}$ is that we have already included generators $H_{AC}$ in case (2) discussed above, which, together with $H_{AB}$ will generate $H_{BC}$ already. In other words, elements in $H_{BC}$ are already included in case (2) above, so we should avoid double counting.
We can decompose \(l_{AB}^{k}{{\mathbb{S}}_{AB}}\) into:

\begin{equation}
l_{AB}^{k}{{\mathbb{S}}_{AB}}=\underbrace{{\mathbb{S}}_{A}\otimes v_1 {\mathbb{S}}_{B}+
    u_2 {\mathbb{S}}_{A} \otimes v_2 {\mathbb{S}}_{B}+u_3 {\mathbb{S}}_{A} \otimes v_3 {\mathbb{S}}_{B}...}_{2^{2b_{AB}+g}\text{ terms}},
\label{eq:SM33}\end{equation}
since $l_{AB}^k$ only acts on subsystem $B$. Similarly, we obtain the following results.

\begin{itemize}
\item If \({{P}_{A}}\ne {{I}_{A}},{{P}_{B}}\ne {{I}_{B}}\), we have:
\begin{equation}
    t_{31}=3\cdot {{2}^{2{{f}_{A}}+2{{f}_{B}}+4{{b}_{AB}}+2g}}\left[ 1+2\cdot {{2}^{-2{{b}_{AB}}-g}}+O({{2}^{-{|A|}+{{b}_{AC}}-{{b}_{AB}}}})\right]
\label{eq:SM34}\end{equation}

\item If \({{P}_{A}}= {{I}_{A}},{{P}_{B}}\ne {{I}_{B}}\), we have:
\begin{equation}
    t_{32}=3\cdot {2}^{4{|A|}+2{{f}_{B}}}\left[1+O({{2}^{-b_{AB}-b_{BC}-g-|B|}})\right]
\label{eq:SM35}\end{equation}

\item If \({{P}_{A}}\ne {{I}_{A}},{{P}_{B}}= {{I}_{B}}\), we have:
\begin{equation}
    t_{33}=0
\label{eq:SM36}\end{equation}

\item If \({{P}_{A}}= {{I}_{A}},{{P}_{B}}= {{I}_{B}}\), we have:
\begin{equation}
    t_{34}=0
\label{eq:SM37}\end{equation}

\end{itemize}

(4) \(l_{AB}^{k} \otimes l_{C}^{k} \in \left\langle L_{AC},\ H_{AC}\right\rangle- H_{AC}\), there are \(N_4:=2^{2b_{AC}+g}-2^g\) terms. As before, since $l_{AB}^k \in H_{AC}$ has already been included in case (2) above, we should exclude them to avoid double counting. We can decompose \(l_{AB}^{k}{{\mathbb{S}}_{AB}}\) into:

\begin{equation}
l_{AB}^{k}{{\mathbb{S}}_{AB}}=\underbrace{{u_1 \mathbb{S}}_{A}\otimes {\mathbb{S}}_{B}+
    u_2 {\mathbb{S}}_{A} \otimes v_2 {\mathbb{S}}_{B}+u_3 {\mathbb{S}}_{A} \otimes v_3 {\mathbb{S}}_{B}...}_{2^{2b_{AB}+g}\text{ terms}}
\label{eq:SM38}\end{equation}
since in this case $l_{AB}^k$ leaves subsystem $B$ and hence $\mathbb{S}_B$ untouched. Depending on the operator content of the Pauli string $P$ on subsystem $A$ and $B$, we obtain the following results.

\begin{itemize}

\item If \({{P}_{A}}\ne {{I}_{A}},{{P}_{B}}\ne {{I}_{B}}\), we have:
\begin{equation}
    t_{41}=3\cdot {{2}^{2{{f}_{A}}+2{{f}_{B}}+4{{b}_{AB}}+2g}}\left[1+2\cdot {{2}^{-2{{b}_{AB}}-g}}+O({{2}^{-{|B|}+{{b}_{BC}}-{{b}_{AB}}}})\right]
\label{eq:SM39}\end{equation}

\item If \({{P}_{A}}= {{I}_{A}},{{P}_{B}}\ne {{I}_{B}}\), we have:
\begin{equation}
    t_{42}=0
\label{eq:SM40}\end{equation}

\item If \({{P}_{A}}\ne {{I}_{A}},{{P}_{B}}= {{I}_{B}}\), we have:
\begin{equation}
    t_{43}=3\cdot {2}^{2{{f}_{A}}+4{|B|}}\left[1+O({{2}^{-b_{AB}-b_{AC}-g-|A|}})\right]
\label{eq:SM41}\end{equation}

\item If \({{P}_{A}}= {{I}_{A}},{{P}_{B}}= {{I}_{B}}\), we have:
\begin{equation}
    t_{44}=0
\label{eq:SM42}\end{equation}

\end{itemize}

(5) Finally,  \(l_{AB}^{k} \otimes l_{C}^{k}\) does not belong to any of the four cases listed above. There are \(N_5:=2^{2b_{AC}+2b_{BC}+2g}-2^{2b_{AC}+g}-2^{2b_{BC}+g}+2^g\) terms. In this case, $l_{AB}^k$ does not leave either of $\mathbb{S}_A$ and $\mathbb{S}_B$ invariant, and the decomposition of  \(l_{AB}^{k}{{\mathbb{S}}_{AB}}\) reads:
\begin{equation}
l_{AB}^{k}{{\mathbb{S}}_{AB}}=\underbrace{{u_1 \mathbb{S}}_{A}\otimes v_1 {\mathbb{S}}_{B}+
    u_2 {\mathbb{S}}_{A} \otimes v_2 {\mathbb{S}}_{B}+u_3 {\mathbb{S}}_{A} \otimes v_3 {\mathbb{S}}_{B}...}_{2^{2b_{AB}+g}\text{ terms}}
\label{eq:SM43}\end{equation}

\begin{itemize}
\item If \({{P}_{A}}\ne {{I}_{A}},{{P}_{B}}\ne {{I}_{B}}\), we have:
\begin{equation}
\begin{aligned}
    t_{51}=3\cdot {{2}^{2{{f}_{A}}+2{{f}_{B}}+4{{b}_{AB}}+2g}}\left( \begin{aligned}
  & 1+2\cdot {{2}^{-2{{b}_{AB}}-g}}+\\
  &+O({{2}^{-{|A|}-{{b}_{AB}}-{{b}_{AC}}-g}})+O({{2}^{-{|B|}-{{b}_{BC}}-{{b}_{AB}}-g}}) \\ 
\end{aligned} \right)
\end{aligned}
\label{eq:SM44}\end{equation}

\item If \({{P}_{A}}= {{I}_{A}},{{P}_{B}}\ne {{I}_{B}}\), we have:
\begin{equation}
    t_{52}=0
\label{eq:SM45}\end{equation}

\item If \({{P}_{A}}\ne {{I}_{A}},{{P}_{B}}= {{I}_{B}}\), we have:
\begin{equation}
    t_{53}=0
\label{eq:SM46}\end{equation}

\item If \({{P}_{A}}= {{I}_{A}},{{P}_{B}}= {{I}_{B}}\), we have:
\begin{equation}
    t_{54}=0
\label{eq:SM47}\end{equation}

\end{itemize}

The \(\overline{Y^{\rm lin}}\) is then expressed in Eq.~(\ref{eq:SM47_1}):

\begin{equation}
\begin{aligned}
    &2^{5|A|+5|B|+5|C|}\cdot 2^{-4|C|} \cdot2^{-|C|+b_{AC}+b_{BC}+g} \cdot(1-\overline{Y^{\rm lin}})\\
    &=t_{11}+N_2t_{21}+N_3t_{31}+N_4t_{41}+N_5t_{51}\\
    &+(4^{|B|}-1)\left[t_{12}+N_2t_{22}+N_3t_{32}+N_4t_{42}+N_5t_{52}\right]\\
    &+(4^{|A|}-1)\left[t_{13}+N_2t_{23}+N_3t_{33}+N_4t_{43}+N_5t_{53}\right]\\
    &+(4^{|A|}-1)(4^{|B|}-1)\left[t_{14}+N_2t_{24}+N_3t_{34}+N_4t_{44}+N_5t_{54}\right]
\end{aligned}
\label{eq:SM47_1}
\end{equation}


In the above summation, the prefactor \(2^{-|C|+b_{AC}+b_{BC}+g}\) comes from the size of each coset \(|l_C^k\mathbb{S}_C|\). We observe that some terms may include contributions of the form \(t_{41}=3\cdot {{2}^{2{{f}_{A}}+2{{f}_{B}}+4{{b}_{AB}}+2g}}[ 1+2\cdot {{2}^{-2{{b}_{AB}}-g}}+O({{2}^{-{|B|}+{{b}_{BC}}-{{b}_{AB}}}})]\), which raises concerns about whether these subleading terms \(O({{2}^{-{|B|}+{{b}_{BC}}-{{b}_{AB}}}})\) can be neglected when the entanglement is large. However, we can verify that the number of \(t_{41}\) is  \( N_4\), while the number of \(t_{51}\) is \(N_5\), and we have \(N_4/N_5\approx2^{-2b_{BC}-g}\). Therefore, the subleading terms contributes \(O(2^{-|B|-b_{BC}-b_{AB}-g})\) to the final result.

\section{An example of the decomposition of the stabilizer group}  \label{sec:2} 

In this section, we give a concrete example of the coset decomposition of a stabilizer group that has been widely used in this work. Consider a 5-qubit GHZ state:  

\begin{equation}
    \frac{1}{\sqrt{2}}\left(|00000\rangle + |11111\rangle \right)
\end{equation}

The elements of the corresponding stabilizer group, denoted as \(\mathbb{S}_{AB}\), are listed in Table~\ref{table:table1}. We partition the system into two subsystems: subsystem \(A\) consisting of the first three qubits and subsystem \(B\) consisting of the remaining two qubits. The reduced density matrices for these subsystems are given by Eq.~(\ref{eq:AppB1}):  

\begin{equation}
    \rho_A=\frac{1}{2^3}({III}_A+{IZZ}_A+{ZIZ}_A+{ZZI}_A),\quad \rho_B=\frac{1}{2^2}({II}_B+{ZZ}_B)
\label{eq:AppB1}
\end{equation}  

Defining the stabilizer subgroups as  

\begin{equation}
\mathbb{S}_A:=\{{III}_A,{IZZ}_A,{ZIZ}_A,{ZZI}_A\}, \quad \mathbb{S}_B:=\{{II}_B,{ZZ}_B\},    
\end{equation}
we observe that \(\mathbb{S}_A\otimes\mathbb{S}_B\) forms a subgroup of \(\mathbb{S}_{AB}\). Consequently, \(\mathbb{S}_{AB}\) can be decomposed into this subgroup and its corresponding cosets, as expressed in Eq.~(\ref{eq:AppB2}):  

\begin{equation}
    \mathbb{S}_{AB}=\mathbb{S}_A\otimes\mathbb{S}_B\ \cup\ a_1\mathbb{S}_A\otimes b_1\mathbb{S}_B\ \cup\ \ a_2\mathbb{S}_A\otimes b_2\mathbb{S}_B\ \cup\ a_3\mathbb{S}_A\otimes b_3\mathbb{S}_B
\label{eq:AppB2}
\end{equation}  

The states stabilized by \(\mathbb{S}_A\) form a stabilizer code defined on subsystem \(A\), specifically the code space spanned by \(\{|000\rangle_A, |111\rangle_A\}\). Similarly, the stabilizer code for subsystem \(B\) is spanned by \(\{|00\rangle_B, |11\rangle_B\}\). The operators \(a_k\) and \(b_k\) correspond to the logical operators of the stabilizer codes on \(A\) and \(B\), respectively. Notably, these logical operators are combined in a manner that ensures all \(a_k \otimes b_k\) commute with one another.

\begin{table}[!tb]
\centering
\begin{tabular}{@{\extracolsep{\fill}}|c|c|c|c|}
\hline
$\quad\quad \mathbb{S}_A \otimes \mathbb{S}_B\quad\quad$ & $XXX\mathbb{S}_A \otimes XX\mathbb{S}_B$ & $IIZ\mathbb{S}_A \otimes ZI\mathbb{S}_B$ & $-XXY\mathbb{S}_A \otimes YX\mathbb{S}_B$ \\
\hline
IIIII & XXXXX & IIZZI & -XXYYX \\
\hline
IZZII & -XYYXX & IZIZI & -XYXYX \\
\hline
ZZIII & -YYXXX & ZZZZI & YYYYX \\
\hline
ZIZII & -YXYXX & ZIIZI & -YXXYX \\
\hline
IIIZZ & -XXXYY & IIZIZ & -XXYXY \\
\hline
IZZZZ & XYYYY & IZIIZ & -XYXXY \\
\hline
ZZIZZ & YYXYY & ZZZIZ & YYYXY \\
\hline
ZIZZZ & YXYYY & ZIIIZ & -YXXXY \\
\hline
\end{tabular}
\caption{Stabilizer group of 5-qubit GHZ state and its decomposition.}
\label{table:table1}
\end{table}

\section{Factorized unitary acting on three subregions}\label{sec:3}
In this section, we consider the effect of a factorized unitary acting on all three subregions of a tripartite state: $U=U_A\otimes U_B \otimes U_C$. In this case, the final result reads: 
\begin{equation}
\begin{aligned}
& \overline{Y^{\rm lin}}=1-4\cdot 2^{-N}\big[1+3\cdot {{2}^{-2{{b}_{AB}}-2{{b}_{AC}}-2{g}}}+3\cdot {{2}^{-2{{b}_{AB}}-2{{b}_{BC}}-2{g}}}+3\cdot {{2}^{-2{{b}_{AC}}-2{{b}_{BC}}-2{g}}}\\
&+\frac{3}{2}\cdot {{2}^{-2{{b}_{AB}}-2{{b}_{AC}}-2{{b}_{BC}}-2{g}}}+\frac{9}{2}\cdot {{2}^{-2{{b}_{AB}}-2{{b}_{AC}}-2{{b}_{BC}}-3{g}}}\big]
\end{aligned}
\label{eq:SM50}\end{equation}
where once again subleading corrections are of the order $O(2^{-|A|}, 2^{-|B|}, 2^{-|C|})$. Similarly to the bipartite result Eq.~(\ref{eq:UaUb}), we find that the resulting state also achieves maximal amount of magic, provided that the subsystems are sufficiently entangled. Interestingly, this doesn't require that the $A$, $B$, and $C$ are pairwise entangled. As Eq.~(\ref{eq:SM50}) suggests, if $A\&B$, $B\&C$ are highly entangled while $A\&C$ are not entangled (i.e. \(b_{AB},b_{BC}\) is large, but \(b_{AC}\) is zero), the average $\overline{Y^{\rm lin}}$ still approaches that of a Haar random state.

As a consistency check, we consider a special case of Eq.~(\ref{eq:SM50}), namely, when only $A\&B$ are entangled, and $C$ is not entangled with $AB$. In this case, we can check that the average $1-\overline{Y^{\rm lin}}$ coincides with a product of Eq.~(\ref{eq:UaUb}) in the main text for the bipartite case and a single Haar random state on $C$, which is consistent with our previous results.

We now detail the calculation of Eq.~(\ref{eq:SM50}). Our strategy parallels that used for the bipartite case. We first decompose the stabilizer group for the system $\mathbb{S}_{ABC}$ as
\begin{equation}
    {\mathbb{S}_{ABC}} = \underbrace{{{\mathbb{S}}_{A}} \otimes {{\mathbb{S}}_{B}} \otimes {{\mathbb{S}}_{C}} \ \cup \ {{a}_{1}}{{\mathbb{S}}_{A}} \otimes {{b}_{1}}{{\mathbb{S}}_{B}} \otimes {{c}_{1}}{{\mathbb{S}}_{C}} \ \cup \ \dots}_{d:={2}^{2b_{AB} + 2b_{AC} + 2b_{BC} + 3g}\ {\rm terms}}
\label{eq:SM16}
\end{equation}
where \(a_k \otimes b_k \otimes c_k \in \langle L_{AB},L_{BC},L_{AC},H_X,H_{AB},H_{BC} \rangle\). Notice, however, since our initial state now factorizes into single-qubit, Bell pair and GHZ states, the logical operators $\{a_k, b_k, c_k\}$ have more structures. For example, the $2^{2b_{AB}+g}$ elements in $\langle L_{AB}, H_{AB}\rangle$ are only supported on subsystems $A$ and $B$, hence there are $d_C := 2^{2b_{AB}+g}$ terms in the coset decomposition~(\ref{eq:SM16}) that have the same $c_k$'s. Similarly, there are $d_A:=2^{2b_{BC}+g}$ cosets with the same $a_k$'s, and $d_B:= 2^{2b_{AC}+g}$ cosets with the same $b_k$'s.

After applying \(U := U_A \otimes U_B \otimes U_C\) and averaging over Haar random unitaries, we obtain:

\begin{equation}
    1-\overline{Y^{\rm lin}} = \frac{1}{{2^{5N}}} \sum\limits_{P \in P_N} \sum\limits_{k_1, k_2, k_3, k_4} \alpha_{P_A}(a_{k_1}, a_{k_2}, a_{k_3}, a_{k_4}) \beta_{P_B}(b_{k_1}, b_{k_2}, b_{k_3}, b_{k_4}) \gamma_{P_C}(c_{k_1}, c_{k_2}, c_{k_3}, c_{k_4}).
\label{eq:SM17}
\end{equation}
Now, depending on the operator contents of the Pauli string $P$ on subsystems $A$, $B$ and $C$ (identity or not), we use Eq.~(\ref{eq:SM6}) and Eq.~(\ref{eq:SM7}) to evaluate the above expression. We give an example here. Consider when \(P_A \neq I_A\), \(P_B \neq I_B\), and \(P_C \neq I_C\), using Eq.~(\ref{eq:SM6}) and Eq.~(\ref{eq:SM7}), we expand Eq.~(\ref{eq:SM17}) as follows:

\begin{equation}
    \begin{aligned}
    & \overline{{{{\text{Tr}}^{4}}}(PU{\mathbb{S}_{ABC}}{{U}^{\dagger}})} \\
    & = 2^{2f_A + 2f_B + 2f_C} \left( 3d^2 + 6d d_A + 6d d_B + 6d d_C + 6d \cdot 2^g \right) \\
    & \quad + O(2^{3f_A + 2f_B + 2f_C - 2|A|} d^2) + O(2^{2f_A + 3f_B + 2f_C - 2|B|} d^2) \\
    & \quad + O(2^{2f_A + 2f_B + 3f_C - 2|C|} d^2) + O(2^{3f_A + 3f_B + 3f_C - 2|A| - 2|B| - 2|C|} d^3) \\
    & = 3 \cdot 2^{2f_A + 2f_B + 2f_C} d \left( d + 2d_A + 2d_B + 2d_C + 2 \cdot 2^g + O(2^{-|A|}) + O(2^{-|B|}) + O(2^{-|C|}) \right),
    \end{aligned}
\label{eq:SM18}
\end{equation}
where $d:= 2^{2b_{AB} + 2b_{AC} + 2b_{BC} + 3g}$ is the total number of cosets. 
The first line can be obtained when only leading order terms, namely, the first line in Eq.~(\ref{eq:SM7}) is considered. Then we need to calculate the following summations over $\{k_1, k_2, k_3, k_4\}$, constrained by products of $\delta$-functions: 
\begin{eqnarray}
    &\sum\limits_{k_1, k_2, k_3, k_4} \delta_2(a_{k_1}, a_{k_2}) \delta_2(a_{k_3}, a_{k_4}) \delta_2(b_{k_1}, b_{k_2}) \delta_2(b_{k_3}, b_{k_4}) \delta_2(c_{k_1}, c_{k_2}) \delta_2(c_{k_3}, c_{k_4})=d^2\\
    &\sum\limits_{k_1, k_2, k_3, k_4} \delta_2(a_{k_1}, a_{k_2}) \delta_2(a_{k_3}, a_{k_4}) \delta_2(b_{k_1}, b_{k_2}) \delta_2(b_{k_3}, b_{k_4}) \delta_2(c_{k_1}, c_{k_3}) \delta_2(c_{k_2}, c_{k_4})=dd_C\\
    &\sum\limits_{k_1, k_2, k_3, k_4} \delta_2(a_{k_1}, a_{k_2}) \delta_2(a_{k_3}, a_{k_4}) \delta_2(b_{k_1}, b_{k_3}) \delta_2(b_{k_2}, b_{k_4}) \delta_2(c_{k_1}, c_{k_4}) \delta_2(c_{k_2}, c_{k_3})=2^gd
\label{eq:SM19}
\end{eqnarray}
This is where we need to be careful with the number of cosets having identical $a_k$'s, $b_k$'s, or $c_k$'s.  The first summation is straightforward, if we fix an arbitrary \(a_{k_1} \otimes b_{k_1} \otimes c_{k_1}\), then we find no more constraints except for \(a_{k_2} \otimes b_{k_2} \otimes c_{k_2}=a_{k_1} \otimes b_{k_1} \otimes c_{k_1}\) and \(a_{k_3} \otimes b_{k_3} \otimes c_{k_3}=a_{k_4} \otimes b_{k_4} \otimes c_{k_4}\), so the summation over $k_1$ and $k_3$ can be chosen independently, leading to $d^2$.

For the second summation, because of the factors \(\delta_2(a_{k_1},a_{k_2}),\delta_2 (b_{k_1},b_{k_2})\), we have \(a_{k_1} \otimes b_{k_1} \otimes c_{k_1}=a_{k_2} \otimes b_{k_2} \otimes c_{k_2}\), similarly we have \(a_{k_3} \otimes b_{k_3} \otimes c_{k_3}=a_{k_4} \otimes b_{k_4} \otimes c_{k_4}\). And considering further the factor \(\delta_2(c_{k_1},c_{k_3})\), we find once we fix \(a_{k_1} \otimes b_{k_1} \otimes c_{k_1}\), there are only \(d_C=2^{2b_{AB}+g}\) choices for \(a_{k_3} \otimes b_{k_3} \otimes c_{k_3}\), thus the summation is \(d d_C\).

For the third summation, once we fix \(a_{k_1} \otimes b_{k_1} \otimes c_{k_1}\), because of the factor \(\delta_2(a_{k_1},a_{k_2})\), we need \(a_{k_2} \otimes b_{k_2} \otimes c_{k_2}\in a_{k_1} \otimes b_{k_1} \otimes c_{k_1} \cdot \left\langle L_{BC},H_{BC}\right\rangle\); because of the factor \(\delta_2(b_{k_1},b_{k_3})\), we can choose \(a_{k_3} \otimes b_{k_3} \otimes c_{k_3}\in a_{k_1} \otimes b_{k_1} \otimes c_{k_1}\cdot \left\langle L_{AC},H_{AC}\right\rangle\); considering further the factor \(\delta_2( c_{k_2},c_{k_3})\), we can only choose \(a_{k_2} \otimes b_{k_2} \otimes c_{k_2}\in a_{k_1} \otimes b_{k_1} \otimes c_{k_1}\cdot \left\langle H_{BC}\right\rangle\). Furthermore, once \(a_{k_2} \otimes b_{k_2} \otimes c_{k_2}\) is fixed, \(a_{k_3} \otimes b_{k_3} \otimes c_{k_3}\) is also determined, because $\delta_2 (c_{k_2},c_{k_3})$ implies that $a_{k_3}\otimes b_{k_3} \otimes c_{k_3}$ can only be obtained by acting on $a_{k_1}\otimes b_{k_1}\otimes c_{k_1}$ with stabilizer group element of the same GHZ state as in $a_{k_2}\otimes b_{k_2}\otimes c_{k_2}$. Finally, $a_{k_4}\otimes b_{k_4} \otimes c_{k_4}$ is also fully fixed by the $\delta$-functions. Therefore, there are only two independent summations in the choices for ${k_1}$ and the GHZ state in $H_{BC}$, leading to the final result $2^g d$.

After similar analysis, we obtain:
\begin{eqnarray}
&&\overline{{{\text{Tr}}^{\text{4}}}(PU{\mathbb{S}_{ABC}}{{U}^{\dagger }})}\nonumber \\
&=& \begin{cases}
   3\cdot{{\text{2}}^{\text{2}{{f}_{A}}+\text{2}{{f}_{B}}+\text{2}{{f}_{C}}}}d({d}+2{{d}_{A}}+2{{d}_{B}}+2{{d}_{C}}+2\cdot {{2}^{g}}),&{{P}_{A}}\ne {{I}_{A}},{{P}_{B}}\ne {{I}_{B}},{{P}_{C}}\ne {{I}_{C}} \\
   {{\text{2}}^{4{|A|}+\text{2}{{f}_{B}}+\text{2}{{f}_{C}}}}(3 {d}_{A}^{2}+6{{d}_{A}}),&{{P}_{A}}={{I}_{A}},{{P}_{B}}\ne {{I}_{B}},{{P}_{C}}\ne {{I}_{C}}\\
   3\cdot{{\text{2}}^{4{|A|}+4{|B|}+\text{2}{{f}_{C}}}} ,&{{P}_{A}}={{I}_{A}},{{P}_{B}}={{I}_{B}},{{P}_{C}}\ne {{I}_{C}}  \\
   {{\text{2}}^{4{|A|}+4{|B|}+4{|C|}}},&{{P}_{A}}={{I}_{A}},{{P}_{B}}={{I}_{B}},{{P}_{C}}={{I}_{C}}  \\
\end{cases}.
\label{eq:SM20}
\end{eqnarray}
The subleading corrections are of the order $O(2^{-|A|}, 2^{-|B|}, 2^{-|C|})$ in each line. Cases not listed in Eq.~(\ref{eq:SM20}) are similar. After combining all terms together, we obtain Eq.~(\ref{eq:SM50}).

\begin{figure}[!tb]
    \centering
    \includegraphics[width=0.6\linewidth]{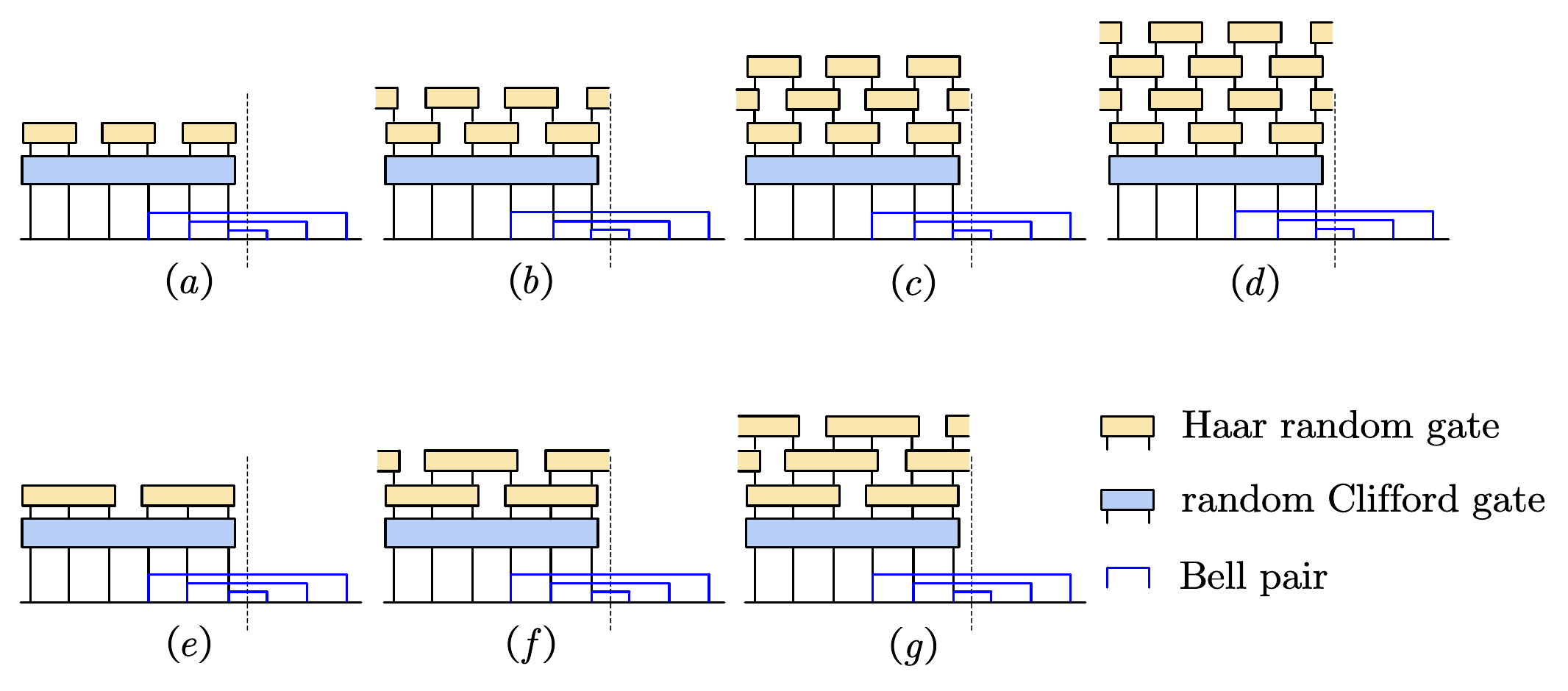}
    \caption{The circuits used to simulate the case where the gate acting on A is not sampled from a global Haar ensemble. The blue lines represent Bell pairs. The blue bricks represent unrelated random Clifford gates, and the yellow bricks represent unrelated random magical gates.}
    \label{fig:SM6}
\end{figure}

\begin{figure}[!tb]
    \centering
    \includegraphics[width=0.5\linewidth]{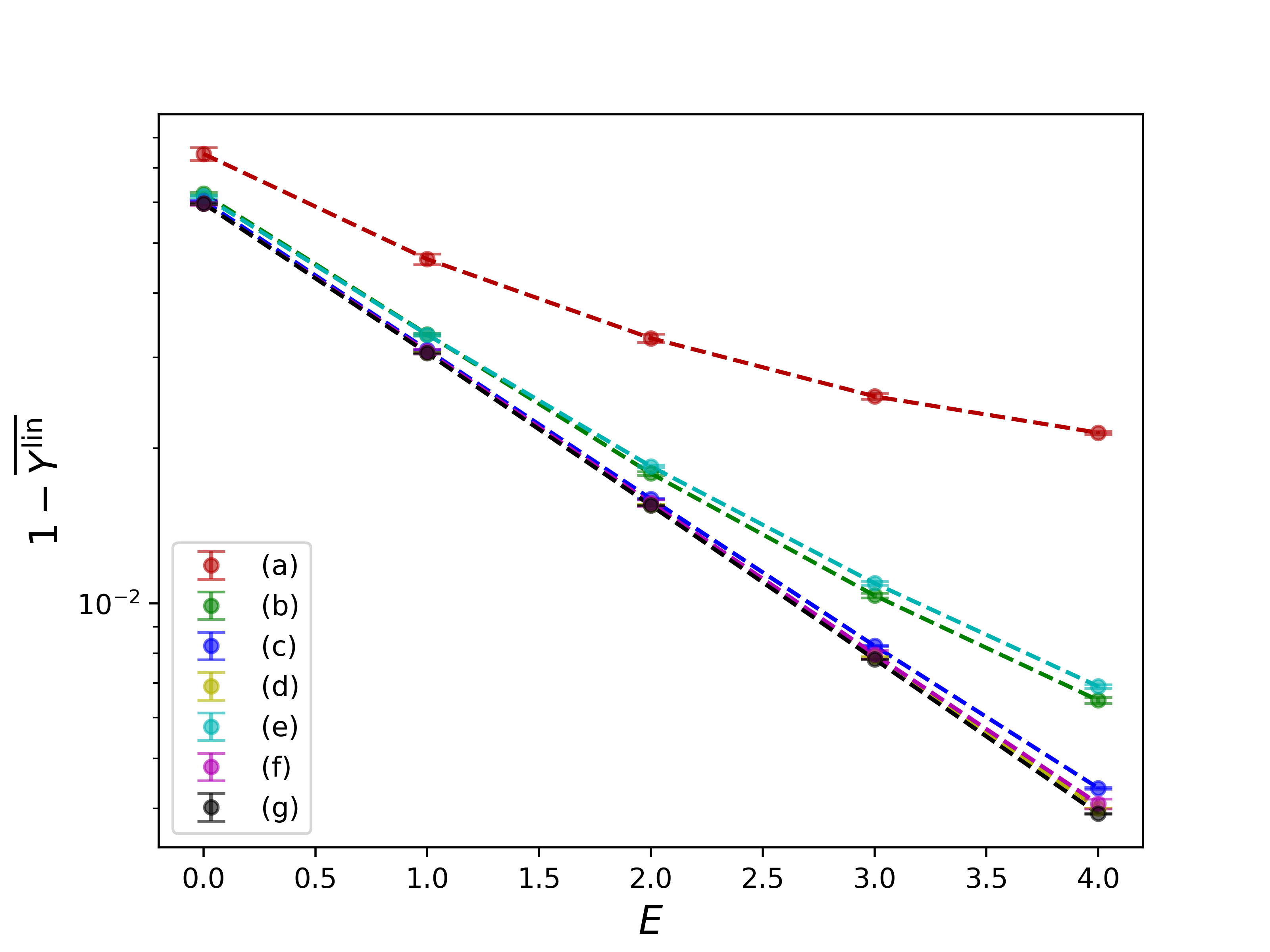}
    \caption{Numerical results for $\overline{Y^{\rm lin}}$ under brickwork circuits of varying depths and gate ranges, as depicted in Fig.~\ref{fig:SM6}.}
    \label{fig:SM7}
\end{figure}

\section{Magic injection with shallow-depth brickwork circuits}\label{sec:4}

In this section, we replace the Haar random unitary gate acting on $A$ with a shallow-depth brickwork circuit. First, we still consider the initial state shown in Fig. 1(d) of the main text, followed by one layer of local gates acting on subregion $A$. Since the initial state factorizes into single-qubit states and Bell-pair states, the average $1-\overline{Y^{\rm lin}}$ also factorizes into a product of contributions from single-qubit and Bell states. Thus, it is easy to see that, \(1-\overline{Y^{\rm lin}}\propto 2^{-E}\) still holds true in this case.

However, if $U_{A}$ is not Haar random, the specific choice of initial state in Fig. 1(d) is no longer the most generic one. To construct the most general stabilizer initial state with tunable entanglement, notice that since any bipartite stabilizer state is local-unitary equivalent to the state in Fig. 1(d), we can simply recover a generic stabilizer state by acting on Fig. 1(d) with random Clifford unitaries $U^{\rm Cliff}_A \otimes U^{\rm Cliff}_B$. Since $U^{\rm Cliff}_B$ commutes with subsequent magic injection gates supported on $A$ and leaves $Y$ invariant, we can simply apply $U^{\rm Cliff}_A$. We then consider brickwork circuits of varying depths and gate ranges, as summarized in Fig.~\ref{fig:SM6}. We then numerically compute the average $\overline{Y^{\rm lin}}$ corresponding to each circuit architecture, and the results are shown in Fig.~\ref{fig:SM7}. We find that the qualitative picture of our analytical result is still valid for the brickwork circuit, in that the average $1-\overline{Y^{\rm lin}}$ decays exponentially with $E$. However, notice that for extremely shallow depths (single layer and two layers), the curve bends upward and saturates, which indicates magic injection that is below capacity. As the circuit depth increases, the behavior of $\overline{Y^{\rm lin}}$ approaches that of a Haar random $U_A$.

\section{Unitary stabilizer R\'enyi entropy}
Here we develop the formalism for unitary magic and T count based on ideas similar to the stabilizer R\'enyi entropy and unitary nullity. 

\begin{definition}
    Let $U$ be an $N$-qubit unitary, then the $\alpha-$th unitary stabilizer R\'enyi entropy of $U$ is defined as

\begin{equation}
        H_{\alpha}(U):=\frac{1}{1-\alpha}\log \Big(\frac{1}{2^{2N}}\sum_{P_{i},P_{j}}\Big[\frac{{\text{tr}}(P_{i}UP_{j}U^{\dagger})}{2^{N}}\Big]^{2\alpha}\Big)
\label{eq:cu}
\end{equation}
where $P_i \in P_N$.
\end{definition}

For a shorthand notation, let \(Q_{\alpha}:=\sum_{P_{i}\in P_N}{P_{i}}^{\otimes2\alpha}\), then we can write it as:
\begin{equation}
    H_{\alpha}(U)=\frac{1}{1-\alpha}\text{log}[{\text{tr}}(Q_{\alpha}U^{\otimes2\alpha}Q_{\alpha}{U^{\dagger\otimes2\alpha}})]-2N\frac{1+\alpha}{1-\alpha}.
\label{eq:SM_G0}
\end{equation}
As explained in the main text, the quantity $H_{\alpha}$ can be viewed as the \(\alpha\)-th R\'enyi entropy of the probability distribution associated with the weight $|c_{ij}|^2$ of $U P_i U^\dagger=\sum_j c_{ij} P_j$ over all $4^N$ Pauli strings. 

\subsection{The relationship between unitary SRE and state SRE}
Consider a $2N$-qubit maximally entangled state between the original system and an auxiliary system, where a unitary $U$ acts on the $N$-qubit subsystem. The unitary stabilizer R\'enyi entropy $H_{\alpha}(U)$ is directly related to the state stabilizer R\'enyi entropy $M_{\alpha}$, defined as:

\begin{equation}
M_{\alpha}(\rho) := 
\frac{1}{1-\alpha}\text{log}[\text{tr}({\rho}^{\otimes2\alpha}Q_{\alpha})]-N\frac{1}{1-\alpha}. 
\end{equation}

\begin{theorem}
    Suppose $U$ is an $N$-qubit unitary operator. Let $|U\rangle_{AB} = U_A\otimes I_B|{{\rm Bell}\rangle}^{\otimes N}_{AB})$ be the Choi state of $U$ where $|A|=|B|=2^N$ and $|\mathrm{Bell}\rangle=(|00\rangle+|11\rangle)/\sqrt{2}$ is maximally entangled between $A$ and $B$ Then $$M_{\alpha}(|U\rangle) = H_{\alpha}(U).$$
\end{theorem}

\begin{proof}
Consider a $2N$-qubit maximally entangled state between the original system $A$ and an auxiliary system $B$ of equal size: $\rho_0=|{\rm Bell}\rangle \langle {\rm Bell}|^{\otimes N}$. Recall that a \(2\)-qubit Bell state can be written as: 
\begin{equation}
    |{\rm Bell}\rangle\langle {\rm Bell}|=\frac{1}{4}(II+XX+ZZ-YY).
\end{equation}
Since the sign in front of each Pauli string in the expansion of $\rho_0$ is not important for calculating \(\text{tr}^{2\alpha}(...)\), below we shall only keep track of the operator content in $\rho_0$ and ignore the sign structure, and write
\begin{equation}
    \rho_0=\frac{1}{2^{2N}}\sum_{P_i \in P_N}\underbrace{P_i}_{A}\underbrace{P_i}_{B}.
\end{equation}
Now we consider the stabilizer R\'enyi entropy of state  \(|U\rangle=U_A\otimes I_B \ |{\rm Bell}\rangle_{AB}^{\otimes N}\):

\begin{equation}
    M_{\alpha}(|U\rangle)=\frac{1}{1-\alpha}\text{log}[\text{tr}({{U_A}^{\otimes 2\alpha}\otimes{I_B}^{\otimes 2\alpha} {\rho_0}^{\otimes 2\alpha} {U_A}^{\dagger\otimes 2\alpha}\otimes{I_B}^{\otimes 2\alpha}}Q^{AB}_{\alpha})]-2N\frac{1}{1-\alpha},
\end{equation}
the last term is not \(-N\frac{1}{1-\alpha}\) because \(|U\rangle\) is a $2N$-qubit state. The key problem is to calculate the quantity inside the logarithm:
\begin{equation}
\begin{aligned}
    &\text{tr}({{U_A}^{\otimes 2\alpha}\otimes{I_B}^{\otimes 2\alpha} {\rho_0}^{\otimes 2\alpha} {U_A}^{\dagger\otimes 2\alpha}\otimes{I_B}^{\otimes 2\alpha}}Q^{AB}_{\alpha})\\=&\sum_{P_{i_1},\dots,P_{i_{2\alpha}}\in P_N}2^{-4\alpha N}\text{tr}(Q^{AB}_{\alpha}{U_A}^{\otimes 2\alpha}\otimes{I_B}^{\otimes 2\alpha} \underbrace{P_{i_1}\cdots P_{i_{2\alpha}}}_{A}\otimes \underbrace{P_{i_1}\cdots P_{i_{2\alpha}}}_{B}{U_A}^{\dagger\otimes 2\alpha}\otimes{I_B}^{\otimes 2\alpha}).
\end{aligned}
\label{eq:SM_G1}
\end{equation}
From the definition of $Q_{\alpha}$, $Q^{AB}_{\alpha}$ can be written as a tensor product \(Q^{AB}_{\alpha}=Q^{A}_{\alpha}\otimes Q^{B}_{\alpha}\), then Eq.~(\ref{eq:SM_G1}) can be written as:

\begin{equation}
\begin{aligned}
    &\text{tr}({{U_A}^{\otimes 2\alpha}\otimes{I_B}^{\otimes 2\alpha} {\rho_0}^{\otimes 2\alpha} {U_A}^{\dagger\otimes 2\alpha}\otimes{I_B}^{\otimes 2\alpha}}Q^{AB}_{\alpha})\\=&
\sum_{P_{i_1},\dots,P_{i_{2\alpha}}\in P_N}2^{-4\alpha N}\text{tr}(Q_{\alpha}{U_A}^{\otimes 2\alpha} \underbrace{P_{i_1}\cdots P_{i_{2\alpha}}}_{A}{U_A}^{\dagger\otimes 2\alpha})\text{tr}(Q_{\alpha}{I_B}^{\otimes 2\alpha} \underbrace{P_{i_1}\cdots P_{i_{2\alpha}}}_{B}{I_B}^{\otimes 2\alpha}),
\end{aligned}
\label{eq:SM_G2}
\end{equation}
because \(Q_\alpha:=\sum_{P_{i}\in P_N}{P_{i}}^{\otimes{2\alpha}}\), it will select those \(P_{i_1}\cdots P_{i_{2\alpha}}\) with \(P_{i_1}=P_{i_2}=\cdots =P_{i_{2\alpha}}\). Then Eq.~(\ref{eq:SM_G2}) further simplifies:
\begin{equation}
\begin{aligned}
    &\text{tr}({{U_A}^{\otimes 2\alpha}\otimes{I_B}^{\otimes 2\alpha} {\rho_0}^{\otimes 2\alpha} {U_A}^{\dagger\otimes 2\alpha}\otimes{I_B}^{\otimes 2\alpha}}Q^{AB}_{\alpha})\\
    =&\sum_{P_{i}\in P_N}2^{-4\alpha N}\text{tr}(Q_{\alpha}{U_A}^{\otimes 2\alpha} \underbrace{P_{i}\cdots P_{i}}_{2\alpha}{U_A}^{\dagger\otimes 2\alpha})\text{tr}(Q_{\alpha}{I_B}^{\otimes 2\alpha} \underbrace{P_{i}\cdots P_{i}}_{2\alpha}{I_B}^{\otimes 2\alpha})\\
    =&2^{-2\alpha N}\text{tr}(Q_{\alpha} {U_A}^{\otimes 2\alpha}Q_{\alpha}{U_A}^{\dagger\otimes 2\alpha}).
\end{aligned}
\label{eq:SM_G3}
\end{equation}
Then \(M_{\alpha}(|U\rangle)\) can be written as:
\begin{equation}
\begin{aligned}
    M_{\alpha}(|U\rangle)=&\frac{1}{1-\alpha}\text{log}[\text{tr}({{U_A}^{\otimes 2\alpha}\otimes{I_B}^{\otimes 2\alpha} {\rho_0}^{\otimes 2\alpha} {U_A}^{\dagger\otimes 2\alpha}\otimes{I_B}^{\otimes 2\alpha}}Q^{AB}_{\alpha})]-2N\frac{1}{1-\alpha}\\
    =&\frac{1}{1-\alpha}\text{log}[2^{-2\alpha N}\text{tr}(Q_{\alpha} {U_A}^{\otimes 2\alpha}Q_{\alpha}{U_A}^{\dagger\otimes 2\alpha)}]-2N\frac{1}{1-\alpha}\\
    =&\frac{1}{1-\alpha}\text{log}[\text{tr}(Q_{\alpha} {U_A}^{\otimes 2\alpha}Q_{\alpha}{U_A}^{\dagger\otimes 2\alpha)}]-2N\frac{1+\alpha}{1-\alpha},
\end{aligned}
\label{eq:SM_G4}
\end{equation}
this is exactly the definition of \(H_{\alpha}(U)\) in Eq.~(\ref{eq:SM_G0}), so we get \(H_{\alpha}(U)=M_{\alpha}(|U\rangle)\).
\end{proof}

We have thus established that the unitary stabilizer R\'enyi entropy $H_{\alpha}(U)$ is equal to the state R\'enyi stabilizer entropy of a maximally entangled state between the original system and an auxiliary system of equal size, acted upon by $U$. Notice that this result does not involve averaging over a particular unitary ensemble, and holds for each and arbitrary $U$. 

\subsection{Lower bound on the unitary nullity and \(T\) count}

In this section, we establish the full inequality of Eq.~(\ref{eq:ineq}) of the main text.
Let \(p_{ij}(U):={\text{tr}}^{2}(P_iUP_jU^{\dagger})/2^{4N}\). Because \(\sum_{i,j}p_{ij}=1\), the coefficients \(p_{ij}\) form a normalized probability distribution. Then the \(\alpha\)-th unitary stabilizer R\'enyi entropy we have defined can be written in the following way:
\begin{equation}
    H_{\alpha}(U):=\frac{1}{1-\alpha}\text{log}\sum_{i,j}{p_{ij}(U)}^{\alpha}-2N.
\end{equation}
To connect the unitary SRE to unitary nullity, we first note the following lemma. 

\begin{theorem}
   Let $t(U)$ be the T-count of $U$, i.e., the number of T gates needed to synthesize $U$ in a Clifford+T decomposition and let $\nu(U)$ be the unitary nullity. Then  $H_0(U)\leq \nu (U)\leq t(U)$.
\end{theorem}

\begin{proof}
   Define \(s(U):=\{P_i| \ UP_iU^\dagger\in P_N\}\) and \(g(U):=Us(U)U^\dagger\). Both \(s(U)\) and \(g(U)\) are subgroups of \(P_N\). Recall the definition of unitary nullity~\cite{Jiang_2023} is \(v(U):=2N-\text{log}|s(U)|\). We will show that, for a fixed \(P_i\), the number of non-zero \(p_{ij}\) is upper bounded by \(\frac{2^{2N}}{|s(U)|}\).

The group \(g(U)\) acts on the operator space via conjugation, \(R(X):=RXR,\ R\in g(U)\). For any \(R\in g(U)\), there exists \(Q \in s(U)\) with \(R=UQU^\dagger\). Thus for any Pauli string \(P_i\) and any $R\in g(U)$, we have:
\begin{equation}
    R(UP_iU^\dagger)=RUP_iU^\dagger R=U(QP_iQ)U^\dagger,
\end{equation}
for some $Q\in s(U)$.
Since Pauli strings either commute or anticommute, that is, $QP_i=\pm P_iQ$, we must have
\begin{equation}
    R(UP_iU^\dagger)=\pm UP_iU^\dagger.
    \label{eq:ru}
\end{equation}
Thus, regarding \(UP_iU^\dagger\) as a vector in operator space and \(R(\cdot)\) as a linear transformation, we see that \(UP_iU^\dagger\) is a common eigenvector of all the \(R\in g(U)\), with eigenvalues \(\pm1\). Expanding \(UP_iU^\dagger\) in Pauli basis,
\begin{equation}
    UP_iU^\dagger=\sum_{S\in P_N}c_SS,
\end{equation}
and applying \(R(\cdot)\) gives
\begin{equation}
    R(UP_iU^\dagger)=\sum_S c_S(RSR)=\pm\sum_Sc_SS.
\end{equation}
Since \(RSR=\pm S\), Eq.~(\ref{eq:ru}) imposes the constraint that all Pauli strings in the above expansion with \(c_S\ne0\) must share the same commutation relation with \(R\). Modulo a global phase, the space of \(P_N\) is isomorphic to the binary symplectic space \(\mathbb{F}_2^{2N}\). Each Pauli string \(S=X_1^{s_1}Z_1^{s_2}X_2^{s_3}Z_2^{s_4}...X_N^{s_{2N-1}}Z_N^{s_{2N}}\) corresponds to a vector \(s=(s_1,s_2,...,s_{2N})\), and each \(R \in g(U)\) can be equivalently represented as a vector \(r \in \mathbb{F}_2^{2N} \). The commutation relation between \(R\) and \(S\) is encoded by the symplectic form:
\begin{equation}
    \varepsilon(R,S)=(-1)^{\langle r,s\rangle},\quad\langle r,s\rangle:=\sum^N_{i=1}r_{2i-1}s_{2i}+r_{2i}s_{2i-1} \quad ({\rm mod} \ 2)
\end{equation}
where the \(\langle r,s\rangle\) is the symplectic inner product of \(r,s\). 
Therefore, the constraint that all nonzero \(c_S\) share the same commutation relation with \(R\in g(U)\) becomes a system of linear equations:

\begin{equation}
 \langle r,s\rangle=b(r) \in \{0,1\} \quad \forall \ r \in G,  
\end{equation}
where \(G \subseteq \mathbb{F}^{2N}_2\) is the vector space corresponding to \(g(U)\), with \(\text{dim}G=\text{log}|g(U)|=\text{log}|s(U)|\). Thus the solution space lies within the affine translation of the symplectic orthogonal complement \(G^{\perp}\), whose size is \(|G^{\perp}|=2^{2N-\text{dim}G}=\frac{2^{2N}}{|s(U)|}\). Hence, the number of nonzero coefficients \(c_S\) is at most \(\frac{2^{2N}}{|s(U)|}\). Summing over all Pauli strings $P_i$, we obtain: 

\begin{equation}
    H_0(U)=\text{log}(\#\{ p_{ij}\ne 0\})-2N\le\text{log}(\frac{2^{4N}}{|s(U)|})-2N=v(U).
\end{equation}
Since \(t(U)\ge v(U)\)~\cite{Jiang_2023}, it follows that \(t(U)\ge v(U) \geq H_0(U)\).
\end{proof}

\begin{corollary}
 For any unitary $U$, its $T$-count is lower bounded such that $t(U) \ge  H_{\alpha}(U)$.
\end{corollary}
\begin{proof}
Recall that  R\'enyi entropy is non-increasing in \(\alpha\) such that
    $H_{\infty}(U)\le H_{\alpha}(U) \le H_0(U)$
where \(H_0(U)=\text{log}(\#\{p_{ij}\ne 0\})-2N\). 
\end{proof}
Although $H_{\alpha}(U)$ is not as tight a lower bound compare to unitary nullity, the calculation of $H_{\alpha}(U)$ does not involve any optimization. 
Furthermore, if $U_A$ has a finer structure, such as a finite depth circuit, then the stabilizer Renyi entropy of its Choi state can be computed efficiently using tensor network methods. Hence we obtain a computable measure of the unitary magic and a computable lower bound of $t(U)$.

\section{Exact results}\label{sec:5}
In this section, we present an alternative analytical method to derive exact expressions of $\overline{Y^{\rm lin}}$ in all cases discussed in the main text\footnote{We acknowledge Alioscia Hamma for providing insightful notes, which served as the primary inspiration for the analytical method presented in this section.}. Recall the linear stabilizer entropy for a pure state $\rho =|\psi\rangle \langle \psi |$ is defined as $Y^{\rm lin}(\rho) :=1-\frac{1}{D}\sum_{P\in P_N}{\rm tr}(P \rho)^4$. where \(D:={\rm dim}\  \mathcal{H}\). This quantity admits an equivalent expression: 
\begin{equation}
    Y^{\rm lin}(\rho)=1-\frac{1}{D}{\rm tr} (Q\rho^{\otimes 4}),
\label{eq:SMp_1}
\end{equation}
where \(Q:=\sum_{P_{i}\in P_N}{P_{i}}^{\otimes4}\). For a Hilbert space \(\mathcal{H}=\otimes^{n}_{i=1}\mathcal{H}_i\), the operator \(Q\) factories as \(Q=\otimes^{n}_{i=1}Q_i\). We can then interpret the map \({\rm tr} (Q(\cdot))\) as a linear operator, which facilitates subsequent calculations.
\subsection{Eq.~(\ref{eq:Y_bi})}
In this subsection, we present a detailed derivation of the exact form of Eq.~(\ref{eq:Y_bi}). The initial state is given by: 
\begin{equation}
    |\psi\rangle=\sum_{i\in\{0,1\}^{\otimes E}} 2^{-E/2}{|0\rangle}^{\otimes f_A}{|i\rangle}_A\otimes{|i\rangle}_B.
\label{eq:SMp_2}
\end{equation}
Let \(\vec{i}:=(i_1,i_2,i_3,i_4)\), where each \(i_k\in\{0,1\}^{\otimes E}\), denote a tuple of four bit strings. Then we have:
\begin{equation}
    {|\psi\rangle\langle\psi|}^{\otimes 4}=2^{-4E}\sum_{\vec{i},\vec{i'}}{|0\rangle\langle0|}^{\otimes 4f_A}\otimes{|\vec{i}\rangle\langle{\vec{i'}}|}_A\otimes{|\vec{i}\rangle\langle\vec{i'}|}_B.
\label{eq:SMp_3}
\end{equation}
Substituting this into the expression for the quantity in Eq.~(\ref{eq:Y_bi}), we have:
\begin{equation}
    \overline{Y^{\rm lin}}
    =\mathbb{E}_{U_A} \ Y^{\rm lin}\left(U_A|\psi\rangle \langle\psi|\ U_A^\dagger\right)
    =1-\frac{1}{D}{\rm tr} \left[Q\mathbb{E}_{U_A}\left({U_A}^{\otimes 4}\otimes {I_B}^{\otimes 4} {|\psi\rangle\langle\psi|}^{\otimes 4} \ {U_A^\dagger}^{\otimes 4}\otimes {I_B}^{\otimes 4}\right)\right].
\label{eq:SMp_4}
\end{equation}
Notice the Haar average of an operator \(O\in\mathcal{B}(\mathcal{H}^{\otimes k})\) reads:
\begin{equation}
    \mathbb{E}_{U}U^{\otimes k}O{U^\dagger}^{\otimes k}=\sum_{\sigma,\pi\in S_k}V_{\sigma,\pi}{\rm tr}(T_{\sigma}O)T_{\pi}
\label{eq:SMp_5},
\end{equation}
with \(V_{\sigma,\pi}\) being the Weingarten function and \(T_{\sigma},T_{\pi}\) being the aforementioned permutation operators. Applying this to Eq.~(\ref{eq:SMp_4}), we have:
\begin{equation}
    1-\overline{Y^{\rm lin}}=\sum_{\sigma,\pi}\sum_{\vec{i},\vec{i'}}\frac{1}{D}2^{-4E}V^A_{\sigma,\pi}{\rm tr} \left(Q_AT^A_{\pi}\right){\rm tr}\left(T^{A-f_A}_{\sigma}{|\vec{i}\rangle\langle{\vec{i'}}|}_A\right){\rm tr}\left(Q_B{|\vec{i}\rangle\langle\vec{i'}|}_B\right).
\label{eq:SMp_6}
\end{equation}
We have used the fact that \(Q=Q_A\otimes Q_B\) and \(T^A_{\sigma}=T^{f_A}_{\sigma}\otimes T^{A-f_A}_{\sigma}\). Furthermore, because of the tensor product structure, we have:
\begin{eqnarray}
    F_1(\sigma,E)&:=&\sum_{\vec{i},\vec{i'}\in {\{0,1\}}^{\otimes4E}}{\rm tr}\left(T^{A-f_A}_{\sigma}{|\vec{i}\rangle\langle{\vec{i'}}|}_A\right){\rm tr}\left(Q_B{|\vec{i}\rangle\langle\vec{i'}|}_B\right)\\\nonumber
    &=&\sum_{\vec{i}\in {\{0,1\}}^{\otimes4E}}{\rm tr}\left(Q_B{|\vec{i}\rangle\langle\sigma(\vec{i})|}_B\right)=\left[\sum_{\vec{j}\in {\{0,1\}}^{\otimes4}}{\rm tr}\left(Q_1{|\vec{j}\rangle\langle\sigma(\vec{j})|}_1\right)\right]^E,
\label{eq:SMp_7}
\end{eqnarray}
where \(\sigma(\vec{i}):=(i_{\sigma(1)},i_{\sigma(2)},i_{\sigma(3)},i_{\sigma(4)})\). Then the Eq.~(\ref{eq:SMp_6}) can be written as the formula shown in Eq.~(\ref{eq:SMp_8}).
\begin{equation}
    1-\overline{Y^{\rm lin}}=\sum_{\sigma,\pi}2^{-4E}DV^A_{\sigma,\pi}{\rm tr} \left(Q_AT^A_{\pi}\right)F_1(\sigma,E)
\label{eq:SMp_8},
\end{equation}
where all the terms represent rational fractions, and the entire summation can be completed using symbolic computation. Let \(D_A:=2^{|A|},D_E:=2^E\). Finally, we have the exact expression Eq.~(\ref{eq:SMp_9}).
\begin{equation}
    \overline{Y^{\rm lin}}=1-\frac{4(D_A^2D_E^2-3D_AD_E-6D_E^2+6)}{D_A(D_A^2-9)D_E^3}.
\label{eq:SMp_9}
\end{equation}

\subsection{Eq.~(\ref{eq:UaUb})}
In this subsection, we present the analytical method to derive the exact expression of Eq.~(\ref{eq:UaUb}). The initial state is shown in Eq.~(\ref{eq:SMp_10}).
\begin{equation}
    |\psi\rangle=\sum_{i\in\{0,1\}^{\otimes E}} 2^{-E/2}{|0\rangle}^{\otimes f_A}{|i\rangle}_A\otimes{|0\rangle}^{\otimes f_B}{|i\rangle}_B,
\label{eq:SMp_10}
\end{equation}
and the tensor product of its density matrix reads:
\begin{equation}
    {|\psi\rangle\langle\psi|}^{\otimes 4}=2^{-4E}\sum_{\vec{i},\vec{i'}}{|0\rangle\langle0|}^{\otimes 4f_A}\otimes{|\vec{i}\rangle\langle{\vec{i'}}|}_A\otimes{|0\rangle\langle0|}^{\otimes 4f_B}\otimes{|\vec{i}\rangle\langle\vec{i'}|}_B.
\label{eq:SMp_11}
\end{equation}
The quantity in Eq.~(\ref{eq:UaUb}) then can be written as:
\begin{equation}
    \overline{Y^{\rm lin}}=\mathbb{E}_{U_A}\mathbb{E}_{U_B} \ Y^{\rm lin}\left(U_AU_B|\psi\rangle \langle\psi|\ U_A^\dagger U_B^\dagger\right)=1-\frac{1}{D}{\rm tr} \left[Q\mathbb{E}_{U_A}\mathbb{E}_{U_B}\left({U_A}^{\otimes 4}\otimes {U_B}^{\otimes 4} {|\psi\rangle\langle\psi|}^{\otimes 4} \ {U_A^\dagger}^{\otimes 4}\otimes {U_B^\dagger}^{\otimes 4}\right)\right].
\label{eq:SMp_12}
\end{equation}
Applying Eq.~(\ref{eq:SMp_5}), we have:
\begin{equation}
    1-\overline{Y^{\rm lin}}=\sum_{\sigma_A,\pi_A,\sigma_B,\pi_B}\sum_{\vec{i},\vec{i'}}2^{-4E}\frac{1}{D}V^A_{\sigma_A,\pi_A}V^B_{\sigma_B,\pi_B}{\rm tr} \left(Q_AT^A_{\pi_A}\right){\rm tr} \left(Q_BT^B_{\pi_B}\right){\rm tr}\left(T^{A-f_A}_{\sigma_A}{|\vec{i}\rangle\langle{\vec{i'}}|}_A\right){\rm tr}\left(T^{B-f_B}_{\sigma_B}{|\vec{i}\rangle\langle{\vec{i'}}|}_B\right).
\label{eq:SMp_13}
\end{equation}
The summation over \(\vec{i},\vec{i'}\) can be done similarly in Eq.~(\ref{eq:SMp_7}):
\begin{eqnarray}
    F_2(\sigma_A,\sigma_B,E)&:=&\sum_{\vec{i},\vec{i'}\in {\{0,1\}}^{\otimes4E}}{\rm tr}\left(T^{A-f_A}_{\sigma_A}{|\vec{i}\rangle\langle{\vec{i'}}|}_A\right){\rm tr}\left(T^{B-f_B}_{\sigma_B}{|\vec{i}\rangle\langle{\vec{i'}}|}_B\right)\\ \nonumber
    &=&\sum_{\vec{i}\in {\{0,1\}}^{\otimes4E}}\delta_{\sigma_A(\vec{i}),\sigma_B(\vec{i})}=\left(\sum_{\vec{j}\in {\{0,1\}}^{\otimes4}}\delta_{\sigma_A(\vec{j}),\sigma_B(\vec{j})}\right)^E.
\label{eq:SMp_14}
\end{eqnarray}
Then we derive a summation formula Eq.~(\ref{eq:SMp_15}), whose terms contains rational fractions, so the entire summation can be completed using symbolic computation.
\begin{equation}
    1-\overline{Y^{\rm lin}}=\sum_{\sigma_A,\pi_A,\sigma_B,\pi_B}2^{-4E}\frac{1}{D}V^A_{\sigma_A,\pi_A}V^B_{\sigma_B,\pi_B}{\rm tr} \left(Q_AT^A_{\pi_A}\right){\rm tr} \left(Q_BT^B_{\pi_B}\right)F_2(\sigma_A,\sigma_B,E).
\label{eq:SMp_15}
\end{equation}
Let \(D_A:=2^{|A|},D_B:=2^{|B|},D_E:=2^E\), we have:
\begin{eqnarray}
    &1-\overline{Y^{\rm lin}}=\frac{4}{D_AD_B(D_A^2-9)(D_B^2-9)D_E^3}\big(D_A^2D_B^2D_E(D_E^2+3)-6D_AD_B(D_A+D_B)(D_E^2+1)\\ \nonumber
    &-6(D_A^2+D_B^2-9)D_E(D_E^2-1)+3D_AD_BD_E(D_E^2+11)\big).
\label{eq:SMp_16}
\end{eqnarray}

\subsection{Eq.~(\ref{eq:tripartite_UaUb})}
In this subsection, we present the analytical method to derive the exact expression of Eq.~(\ref{eq:tripartite_UaUb}). The initial state is shown in Eq.~(\ref{eq:SMp_17}).

\begin{equation}
    |\psi\rangle=2^{-(b_{AB}+b_{AC}+b_{BC}+g)/2}\sum_{i_{AB},i_{AC},i_{BC},i_g} {|0\rangle}^{\otimes f_A}{|i_{AB}\rangle}{|i_{AC}\rangle}{|i_{g}\rangle}_A\otimes{|0\rangle}^{\otimes f_B}{|i_{AB}\rangle}{|i_{BC}\rangle}{|i_{g}\rangle}_B\otimes{|i_{AC}\rangle}{|i_{BC}\rangle}{|i_{g}\rangle}_C.
\label{eq:SMp_17}
\end{equation}
The tensor product of its density matrix reads:
\begin{eqnarray}
    &{|\psi\rangle\langle\psi|}^{\otimes4}=2^{-4(b_{AB}+b_{AC}+b_{BC}+g)}\sum_{\vec{i}_{AB},\vec{i}'_{AB},\vec{i}_{AC},\vec{i}'_{AC},\vec{i}_{BC},\vec{i}'_{BC},\vec{i}_g,\vec{i}'_{g}} {{|0\rangle\langle0|}^{\otimes 4f_A}{|\vec{i}_{AB}\rangle\langle\vec{i}'_{AB}|\otimes}{|\vec{i}_{AC}\rangle\langle\vec{i}'_{AC}|}\otimes|\vec{i}_{g}\rangle\langle\vec{i}'_g|}_A\\ \nonumber
    &\otimes{{|0\rangle\langle0|}^{\otimes 4f_B}{|\vec{i}_{AB}\rangle\langle\vec{i}'_{AB}|\otimes}{|\vec{i}_{BC}\rangle\langle\vec{i}'_{BC}|}\otimes|\vec{i}_{g}\rangle\langle\vec{i}'_g|}_B\otimes{{|\vec{i}_{AC}\rangle\langle\vec{i}'_{AC}|\otimes}{|\vec{i}_{BC}\rangle\langle\vec{i}'_{BC}|}\otimes|\vec{i}_{g}\rangle\langle\vec{i}'_g|}_C.
\label{eq:SMp_18}
\end{eqnarray}
The quantity in Eq.~(\ref{eq:tripartite_UaUb}) then can be written as:
\begin{eqnarray}
    \overline{Y^{\rm lin}}&=&\mathbb{E}_{U_A}\mathbb{E}_{U_B} \ Y^{\rm lin}\left(U_AU_B|\psi\rangle \langle\psi|\ U_A^\dagger U_B^\dagger\right)\\ \nonumber
    &=&1-\frac{1}{D}{\rm tr} \left[Q\mathbb{E}_{U_A}\mathbb{E}_{U_B}\left({U_A}^{\otimes 4}\otimes {U_B}^{\otimes 4}\otimes{I_B}^{\otimes 4} {|\psi\rangle\langle\psi|}^{\otimes 4} \ {U_A^\dagger}^{\otimes 4}\otimes {U_B^\dagger}^{\otimes 4}\otimes{I_C}^{\otimes 4}\right)\right].
\label{eq:SMp_19}
\end{eqnarray}
Applying Eq.~(\ref{eq:SMp_5}), we have:
\begin{eqnarray}
    &1-\overline{Y^{\rm lin}}=2^{-4(b_{AB}+b_{AC}+b_{BC}+g)}\frac{1}{D}\sum_{\sigma_A,\pi_A,\sigma_B,\pi_B}\sum_{\vec{i}_{AB},\vec{i}'_{AB},\vec{i}_{AC},\vec{i}'_{AC},\vec{i}_{BC},\vec{i}'_{BC},\vec{i}_g,\vec{i}'_{g}}V^A_{\sigma_A,\pi_A}V^B_{\sigma_B,\pi_B}\times\\ \nonumber
    &{\rm tr}\left(T^{A-f_A}_{\sigma_A}{{|\vec{i}_{AB}\rangle\langle\vec{i}'_{AB}|\otimes}{|\vec{i}_{AC}\rangle\langle\vec{i}'_{AC}|}\otimes|\vec{i}_{g}\rangle\langle\vec{i}'_g|}_A\right){\rm tr}\left(T^{B-f_B}_{\sigma_B}{{|\vec{i}_{AB}\rangle\langle\vec{i}'_{AB}|\otimes}{|\vec{i}_{BC}\rangle\langle\vec{i}'_{BC}|}\otimes|\vec{i}_{g}\rangle\langle\vec{i}'_g|}_B\right)\times\\ \nonumber
    &{\rm tr} \left(Q_AT^A_{\pi_A}\right){\rm tr} \left(Q_BT^B_{\pi_B}\right){\rm tr}\left(Q_C{{|\vec{i}_{AC}\rangle\langle\vec{i}'_{AC}|\otimes}{|\vec{i}_{BC}\rangle\langle\vec{i}'_{BC}|}\otimes|\vec{i}_{g}\rangle\langle\vec{i}'_g|}_C\right).
\label{eq:SMp_20}
\end{eqnarray}
The summation over \(\vec{i}_{AB},\vec{i}'_{AB},\vec{i}_{AC},\vec{i}'_{AC},\vec{i}_{BC},\vec{i}'_{BC},\vec{i}_g,\vec{i}'_{g}\) can be done in the similar way as those Eq.~(\ref{eq:SMp_7}) and Eq.~(\ref{eq:SMp_14}).
\begin{eqnarray}
    && F_3(\sigma_A,\sigma_B,b_{AB},b_{AC},b_{BC},g)\\ \nonumber
    &:=&\sum_{\vec{i}_{AB},\vec{i}'_{AB},\vec{i}_{AC},\vec{i}'_{AC},\vec{i}_{BC},\vec{i}'_{BC},\vec{i}_g,\vec{i}'_{g}}{\rm tr}\left(Q_C{{|\vec{i}_{AC}\rangle\langle\vec{i}'_{AC}|\otimes}{|\vec{i}_{BC}\rangle\langle\vec{i}'_{BC}|}\otimes|\vec{i}_{g}\rangle\langle\vec{i}'_g|}_A\right)\times\\ \nonumber
    &\times&{\rm tr}\left(T^{A-f_A}_{\sigma_A}{{|\vec{i}_{AB}\rangle\langle\vec{i}'_{AB}|\otimes}{|\vec{i}_{AC}\rangle\langle\vec{i}'_{AC}|}\otimes|\vec{i}_{g}\rangle\langle\vec{i}'_g|}_A\right){\rm tr}\left(T^{B-f_B}_{\sigma_B}{{|\vec{i}_{AB}\rangle\langle\vec{i}'_{AB}|\otimes}{|\vec{i}_{BC}\rangle\langle\vec{i}'_{BC}|}\otimes|\vec{i}_{g}\rangle\langle\vec{i}'_g|}_B\right)\\ \nonumber
    &=&\sum_{\vec{i}_{AB},\vec{i}_{AC},\vec{i}_{BC},\vec{i}_g}{\rm tr}\left(Q_{AC}|\vec{i}_{AC}\rangle\langle\sigma_A(\vec{i}_{AC})|\right){\rm tr}\left(Q_{BC}{|\vec{i}_{BC}\rangle\langle\sigma_B(\vec{i}_{BC})|}\right){\rm tr}\left(Q_g|\vec{i}_{g}\rangle\langle\sigma_A(\vec{i}_g)|\right)\delta_{\sigma_A(\vec{i}_{AB}),\sigma_B(\vec{i}_{AB})}\delta_{\sigma_A(\vec{i}_g),\sigma_B(\vec{i}_g)}\\\nonumber
    &=&\left(\sum_{\vec{j}}\delta_{\sigma_A(\vec{j}),\sigma_B(\vec{j})}\right)^{b_{AB}} \left[\sum_{\vec{j}}{\rm tr}\left(Q_1|\vec{j}\rangle\langle\sigma_A(\vec{j})|\right)\right]^{b_{AC}}\left[\sum_{\vec{j}}{\rm tr}\left(Q_1|\vec{j}\rangle\langle\sigma_B(\vec{j})|\right)\right]^{b_{BC}}\left[\sum_{\vec{j}}{\rm tr}\left(Q_1|\vec{j}\rangle\langle\sigma_B(\vec{j})|\right)\delta_{\sigma_A(\vec{j}),\sigma_B(\vec{j})}\right]^{g}.
\label{eq:SMp_21}
\end{eqnarray}
Then we derive a summation formula Eq.~(\ref{eq:SMp_22}), whose terms contains rational fractions, so the entire summation can be completed using symbolic computation.
\begin{equation}
    1-\overline{Y^{\rm lin}}=2^{-4(b_{AB}+b_{AC}+b_{BC}+g)}\frac{1}{D}\sum_{\sigma_A,\pi_A,\sigma_B,\pi_B}V^A_{\sigma_A,\pi_A}V^B_{\sigma_B,\pi_B}{\rm tr} \left(Q_AT^A_{\pi_A}\right){\rm tr} \left(Q_BT^B_{\pi_B}\right)F_3(\sigma_A,\sigma_B,b_{AB},b_{AC},b_{BC},g).
\label{eq:SMp_22}
\end{equation}
Let \(D_A:=2^{|A|},D_B:=2^{|B|},D_C:=2^{|C|},D_{AB}:=2^{b_{AB}},D_{AC}:=2^{b_{AC}},D_{BC}:=2^{b_{BC}},D_g:=2^g\), we have:

\begin{eqnarray}
    \nonumber 1-\overline{Y^{\rm lin}}&=&\frac{4}{D_AD_B(D_A^2-9)(D_B^2-9)D_{AB}^3D_{AC}^3D_{BC}^3D_g^3}\big(D_A^2D_B^2D_{AB}^3D_{AC}^2D_{BC}^2D_g^2+3D_A^2D_B^2D_{AB}D_{AC}^2D_{BC}^2D_g \\ \nonumber
&-& 6 D_A^2 D_{AB}^3 D_{AC}^2 D_{BC}^2 D_g^2
- 6 D_A^2 D_{AB}^2 D_{AC}^2 D_B D_{BC} D_g 
- 18 D_A^2 D_{AB} D_{AC}^2 D_{BC}^2 D_g 
+ 24 D_A^2 D_{AB} D_{AC}^2 \\ \nonumber
&-& 6 D_A^2 D_{AC}^2 D_B D_{BC} 
+ 3 D_A D_{AB}^3 D_{AC} D_B D_{BC} D_g 
- 6 D_A D_{AB}^2 D_{AC} D_B^2 D_{BC}^2 D_g 
+ 36 D_A D_{AB}^2 D_{AC} D_{BC}^2 D_g \\ \nonumber
&- &36 D_A D_{AB}^2 D_{AC} 
+ 3 D_A D_{AB} D_{AC} D_B D_{BC} D_g 
+ 30 D_A D_{AB} D_{AC} D_B D_{BC} 
- 6 D_A D_{AC} D_B^2 D_{BC}^2 \\ \nonumber
&+& 36 D_A D_{AC} D_{BC}^2 
- 36 D_A D_{AC} 
- 6 D_{AB}^3 D_{AC}^2 D_B^2 D_{BC}^2 D_g^2 
+ 36 D_{AB}^3 D_{AC}^2 D_{BC}^2 D_g^2 
+ 18 D_{AB}^3 \\ \nonumber
&+& 36 D_{AB}^2 D_{AC}^2 D_B D_{BC} D_g 
- 36 D_{AB}^2 D_B D_{BC} 
- 18 D_{AB} D_{AC}^2 D_B^2 D_{BC}^2 D_g 
+ 108 D_{AB} D_{AC}^2 D_{BC}^2 D_g \\
&-&144 D_{AB} D_{AC}^2 
+ 24 D_{AB} D_B^2 D_{BC}^2 
- 144 D_{AB} D_{BC}^2
+ 126 D_{AB} 
+ 36 D_{AC}^2 D_B D_{BC} 
- 36 D_B D_{BC}\big).
\label{eq:SMp_23}
\end{eqnarray}

\section{Simulation method}\label{sec:6}
In this section, we briefly discuss the numerical simulation method we used. We use two different methods for numerically evaluating the average $\overline{Y^{\rm lin}}$. The first one is a brute-force evaluation according to the definition of $Y^{\rm lin}$. However, this requires a direct enumeration of all $4^N$ Pauli strings, which can be time-consuming for larger system sizes. We thus adopt a second method, which is inspired by our analytical calculations to avoid summing over exponentially many Pauli strings. According to Eq.~(\ref{eq:SM10}), there are only four different values for \({{\alpha }_{{{P}_{A}}}}({{a}_{k}},{{a}_{k}},{{a}_{k}},{{a}_{k}})\), depending on whether $P_A = I_A$ and $a_k = I_A$. Therefore, we only need to evaluate the corresponding $\alpha$'s to efficiently compute the average $\overline{Y^{\rm lin}}$. Similarly, there are four possible contributions from different \(P\) in Eq.~(\ref{eq:SM12}), so we can simply simulate these values to obtain an unbiased estimate for \(\overline{Y^{\rm lin}}\).

\end{document}